%% file: goal_oriented_oed.tex
\documentclass[10pt]{iopart}

\usepackage[colorlinks,bookmarksopen,bookmarksnumbered,citecolor=red,urlcolor=red]{hyperref}
\usepackage{algorithm}
\usepackage{algpseudocode}
\usepackage{graphicx}
\usepackage{subfigure}
\usepackage{rotating}
\usepackage{multirow}
\usepackage{fullpage}
\usepackage{appendix}
\usepackage{wrapfig}
\usepackage{datetime}

\usepackage{cite}

\newcommand\SubAlg[1]{%
\vspace*{-.7\baselineskip}\Statex\hspace*{\dimexpr-\algorithmicindent-2pt\relax}\rule{0.75\textwidth}{0.4pt}%
\Statex\hspace*{-\algorithmicindent}\textbf{#1}%
\vspace*{-.7\baselineskip}\Statex\hspace*{\dimexpr-\algorithmicindent-2pt\relax}\rule{0.75\textwidth}{0.4pt}%
}

\input{notation}

\begin{document}

\title[Goal-Oriented Optimal Design of Experiments (GOODE)]{Goal-Oriented Optimal Design of Experiments for \\Large-Scale Bayesian Linear Inverse Problems}


\author{Ahmed Attia \footnote{Corresponding author: {Ahmed Attia, \href{mailto:attia@mcs.anl.gov}{attia@mcs.anl.gov}}} }   
	\address{Mathematics and Computer Science Division, Argonne National Laboratory, Argonne, IL}
	\address{Statistical and Applied Mathematical Science Institute (SAMSI), RTP, NC}
	\ead{attia@mcs.anl.gov}

\author{Alen Alexanderian}
	\address{Department of Mathematics, North Carolina State University, Raleigh, NC}
	\ead{alexanderian@ncsu.edu}

\author{Arvind K. Saibaba}
	\address{Department of Mathematics, North Carolina State University, Raleigh, NC}
	\ead{asaibab@ncsu.edu}

\vspace{10pt}
\begin{indented}
	\item[]December 2017
\end{indented}

\begin{abstract} We develop a framework for goal-oriented optimal design of
experiments (GOODE) for large-scale Bayesian linear inverse problems governed
by PDEs. This framework differs from classical Bayesian optimal design of
experiments (ODE) in the following sense: we seek experimental designs that
minimize the posterior uncertainty in the experiment end-goal, e.g., a
 quantity of interest (QoI), rather than the estimated parameter
itself.  This is suitable for scenarios in which the solution of an inverse
problem is an intermediate step and the estimated parameter is then used to
compute a  QoI.  In such problems, a GOODE approach has two benefits:
the designs can avoid wastage of experimental resources by a targeted
collection of data, and the resulting design criteria are computationally
easier to evaluate due to the often low-dimensionality of the QoIs.  We
present two modified design criteria, A-GOODE and D-GOODE, which are natural
analogues of classical Bayesian A- and D-optimal criteria. We analyze the
connections to other ODE criteria, and provide interpretations for the GOODE
criteria  by using tools from information theory.  Then, we develop an
efficient gradient-based optimization framework for solving the GOODE
optimization problems.  Additionally, we present comprehensive numerical
experiments testing the various aspects of the presented approach. The driving
application is the optimal placement of sensors to identify the source of
contaminants in a diffusion and transport problem.  We enforce sparsity of the
sensor placements using an $\ell_1$-norm penalty approach, and propose a
practical strategy for specifying the associated penalty parameter.

\end{abstract}

\vspace{2pc}

\submitto{\IP}

%
%

\section{Introduction} \label{sec:Introduction}
%
Continuous advances in numerical methods and computational
technology have made it feasible to simulate large-scale physical phenomena
such as weather systems, computer vision, and medical imaging. Mathematical
models are widely used in practice to predict the behavioral patterns of such
physical processes. In the applications we consider, the mathematical models
are typically described by systems of partial differential equations (PDEs). However, parameters
that are needed for a full description of the mathematical models, such as
initial and boundary conditions, or coefficients, are typically unknown and
need to be inferred from experimental data by solving an inverse problem. The
acquisition of data is usually a laborious or expensive process
and has a certain cost associated with it. 
Due to budgetary or physical considerations,
often times, only a limited amount of data can be collected. Even in
applications where collecting data is relatively cheap, processing large
amounts of data can be computationally cumbersome, or a poor design may lead to
wastage of resources, or may miss out on important information regarding the
parameters of interest. Therefore, it is important to control the experimental
conditions for data acquisition in a way that makes optimal use of resources to
accurately reconstruct or infer the parameters of interest.  This is known as
Optimal Design of Experiments (ODE). 

In this article we adopt the Bayesian approach for solving inverse problems.
The Bayesian approach has the following ingredients: the data, the mathematical
model, the statistical description of the observational noise, and the prior
information about the parameters we wish to infer. Bayes' theorem is used to
combine these ingredients to produce the posterior distribution, which
encapsulates the uncertainty in every aspect of the inverse problem. The
posterior distribution can be interrogated in various ways: one can compute the
peak of this distribution, called the maximum a posteriori probability (MAP)
estimate, which estimates the posterior mode of the parameter, draw samples
from this distributions, or compute the conditional mean. The Bayesian approach
to ODE aims to minimize various measures of uncertainty in the inferred
parameters by minimizing certain criteria based on the posterior distribution.
Popular examples of design criteria include the Bayesian A- and D-optimal
criteria~\cite{AtkinsonDonev92,ChalonerVerdinelli95,Pukelsheim93}.  For
Gaussian posteriors, the A- and D-optimal criteria are defined as the trace and
the log-determinant of the posterior covariance matrix, respectively.

ODE is an active area of
research~\cite{Ucinski05,AtkinsonDonev92,Pukelsheim93,Pazman86,ChalonerVerdinelli95,
BauerBockKorkelEtAl00,KorkelKostinaBockEtAl04,
HaberHoreshTenorio10,HoreshHaberTenorio10,ChungHaber12,
HuanMarzouk13,LongScavinoTemponeEtAl13,SanduCioacaRao13,HuanMarzouk14,
AlexanderianPetraStadlerEtAl14,LongMotamedTempone15,Ucinski15,
AlexanderianPetraStadlerEtAl16,AlexanderianGloorGhattas16,
BisettiKimKnioEtAl16,CrestelAlexanderianStadlerEtAl17,YuZavalaAnitescu17,
WalshWildeyJakeman17,ruthotto2017optimal}.  
Specifically, in the recent years a number of advances have been made
on optimal design of experiments for large scale 
applications. The articles~\cite{HaberHoreshTenorio08,HaberHoreshTenorio10,HoreshHaberTenorio10,HaberMagnantLuceroEtAl12,TenorioLuceroBallEtAl13,
AlexanderianPetraStadlerEtAl14,
AlexanderianPetraStadlerEtAl16,CrestelAlexanderianStadlerEtAl17} target A-optimal 
experimental designs for large-scale inverse problems.  Fast algorithms for
computing D-optimal experimental design criterion, given by expected
information gain, for nonlinear inverse problems, were introduced
in~\cite{LongScavinoTemponeEtAl13,LongMotamedTempone15,BeckDiaEspathEtAl17}.
An efficient greedy algorithm for computing Bayesian D-optimal designs,
with correlated observations, is introduced in~\cite{khodja2010guided}.
Choosing a D-optimal experimental design that targets a specific region
in the parameter space is discussed in~\cite{djikpesse2012bayesian}, 
where the optimal design minimizes the marginal uncertainties in the region of interest. 
The works~\cite{AlexanderianGloorGhattas16} and~\cite{AlexanderianSaibaba17}
address theory and computational methods for Bayesian D-optimal design in
infinite-dimensional Bayesian linear inverse problems. Motivated by goal-oriented approaches for parameter dimensionality reduction~\cite{LiebermanWillcox13,LiebermanWillcox14,spantini2017goal}, our paper presents theory and methods
for goal-oriented optimal design of experiments (GOODE).

There are two potential drawbacks in the standard Bayesian approach for ODE.
First, in certain applications, what may be of interest is not the
reconstructed parameter in itself, but some prediction quantity involving the reconstructed
parameter. In this situation, it may be desirable to deploy valuable resources
to collect experimental data so as to minimize the uncertainty in the
end-goal, i.e. prediction, rather than the reconstructed quantity. Second, the reconstructed
parameters are often spatial images or infinite dimensional functions. When
discretized on a fine-scale grid, the resulting parameter dimension is very
high. The posterior covariance matrix 
is also very high dimensional; forming and storing this
covariance matrix explicitly is computationally infeasible on problems
discretized on very fine grid resolutions. Consequently, evaluating
the optimal design criteria is challenging. Randomized matrix 
methods~\cite{AvronToledo11,saibaba2016randomized}
have been instrumental in addressing such computational challenges. 
When the dimension of the
predictions is smaller than the dimension of the reconstructed parameter,
working in the prediction space may be computationally beneficial. These two
reasons---the need for designs tailored to predictions and the computational 
savings offered by targeting low-dimensional prediction quantities---motivate us to 
propose goal-oriented criteria, and devise efficient 
algorithms for their computation and optimization.   

As a motivating application, consider the transport of a contaminant
in an urban environment. The inverse problem of interest here seeks to
identify the source of the contaminant from measurements collected at sensor
locations. The standard Bayesian approach to ODE involves controlling the
sensor locations in order to reconstruct the source (represented as a spatial
function) with minimized uncertainty over the 
entire domain. On the other hand, if instead of determining the initial
condition, the goal is to predict the average contaminant around a building
after a certain amount of time has elapsed, then the experimental design should
explicitly account for this goal.
This application is explored in detail in Sections~\ref{sec:Experiment_Setup}
and~\ref{sec:Numerical_Results}. As a preview, in
Figure~\ref{fig:scalar_GOOED_A_Optimal_weights} we show optimal sensor placements
corresponding to three different goals:
roughly speaking, the left, middle, and right figures depict sensor placements that are
focused on the first building, the second building, and both buildings, respectively.
This shows immediately, that incorporating the end goal, i.e., the 
target prediction, in the ODE problem results in different sensor placements, 
which may be valuable in practice. 
More details are provided in Section~\ref{sec:Numerical_Results}.
\begin{figure}[ht]\centering
  \includegraphics[width=0.60\linewidth]{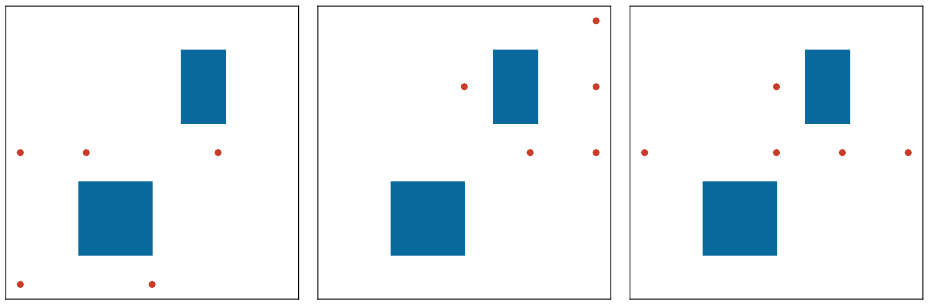}
\caption{GOODE results for  contaminant source identification.  
The blue rectangles model buildings in the physical domain.
The goal is to predict the
average concentration of the contaminant around
left building (left), right building (middle), and both buildings (right) at
prediction time $t_\pred=1.0$. The red dots indicate the respective optimal sensor placements.}
\label{fig:scalar_GOOED_A_Optimal_weights}
\end{figure}

The main contributions of this article are as follows. We propose two goal
oriented criteria---A-GOODE and D-GOODE---that are analogues of the A- and
D-optimality criteria in classical Bayesian ODE.  We investigate the properties
of these two criteria by explaining connections to the classical Bayesian ODE
criteria, and provide interpretations for the criteria by using tools from
information theory. Then, we propose a gradient-based optimization framework for
GOODE using these two criteria. To facilitate this, we derive expressions for
the gradient of these objective functions, along with a computational recipe
for computing them, for large-scale linear inverse problems governed by time-dependent
PDEs. Our proposed strategy is implemented within the context of optimal sensor
placement for the class of inverse problems under study.
We enforce sparsity of the sensor placements using an $\ell_1$-norm
penalty approach, and propose a practical strategy for specifying the
associated penalty parameter.  In addition, we present comprehensive numerical
results, in context of initial state inversion in a time-dependent
advection-diffusion equation, to test the performance of the criteria and the
proposed algorithms. The specific model problem under study is motivated by
the inverse problem of contaminant source identification. 

This article is organized as follows.  Section~\ref{sec:Background} presents a
brief overview of Bayesian inverse problems and optimal design of experiments.
Section~\ref{sec:GOODE} formulates the goal-oriented optimal design of
experiments problem, and presents the computational methods for solving such
problems.  Section~\ref{sec:Experiment_Setup} describes the model inverse
advection-diffusion problem and the setup of numerical experiments used to test
the proposed methods. Numerical results are detailed in
Section~\ref{sec:Numerical_Results}. Concluding remarks are given in
Section~\ref{sec:Conclusions}.

\section{Background}
\label{sec:Background}
In this section, we briefly review the Bayesian formulation of an inverse
problem, and some basics from Bayesian optimal design of experiments.

\subsection{Bayesian inverse problem}\label{subsec:Bayesian_inversion}
Consider the problem of reconstructing an
unknown parameter $\paramcont$ using noisy data $\obs$ and a model
$\mathcal{F}$
\begin{equation}
	\obs  = \mathcal{F}(\paramcont) + \vec{\delta},
\end{equation}
where $\vec{\delta} \in \Rnum^{\Nobs}$ is a centered random
variable that models measurement noise.  Specifically, 
we consider a Gaussian noise model 
$\vec{\delta} \sim \GM{\vec{0}}{\Cobsnoise}$.
We consider the case that $\obs$ is a
(finite-dimensional) vector of measurement data, $\paramcont$ is an element of
an appropriate infinite-dimensional real separable Hilbert space, and
$\mathcal{F}$ is a continuous linear transformation, which will we will refer
to as the parameter-to-observable map. In applications we target $\paramcont
\in L^2(\D)$, where $\D$ is a bounded domain in $\Rnum^n$, with $n = 1,
2$, or $3$.  The infinite-dimensional formulation, finite-element
discretization, and numerical solution of this problem have been addressed in
detail in~\cite{Bui-ThanhGhattasMartinEtAl13}, under the assumption of a
Gaussian prior and additive Gaussian noise model, which is the setting we
consider here.

To keep the presentation simple, we consider the discretized version of the
problem. However, in what follows, we pay close attention to the issues pertaining
to discretization of the infinite-dimensional problem and the infinite-dimensional limit.
We denote the Gaussian prior for the discretized parameter $\iparam
\in \Rnum^{\Nparam}$ as $\GM{\iparb}{\Cparampriormat}$, and assume $\vec{\delta} \sim \GM{\vec{0}}{\Cobsnoise}$. Letting
$\F$ be the discretized parameter-to-observable map, 
the additive Gaussian noise assumption leads to the Gaussian likelihood,
\begin{equation}
\Prob (\obs | \iparam ) \propto \exp{\left( - \frac{1}{2} \sqwnorm{ \F\iparam - \obs }{ \Cobsnoise^{-1} } \right) }.
\end{equation}
That is, $\obs | \iparam \sim \GM{\F\iparam}{\Cobsnoise}$.
It is well known (see e.g.,~\cite[Chapter 3]{Tarantola05}) that, in this setting, the 
posterior is also a Gaussian 
$\GM{\ipara}{\Cparampostmat}$ with
\begin{equation}\label{eqn:Analysis_param_covar}
	\Cparampostmat = \left( \F^* \Cobsnoise^{-1} \F  + \Cparampriormat^{-1} \right)^{-1} \,,  \qquad
	\ipara = \Cparampostmat \left( \Cparampriormat^{-1} \iparb + \F^* \Cobsnoise^{-1}\, \obs \right).
\end{equation}
It is also worth mentioning that the posterior mean $\ipara$ is the minimizer of the following functional
\begin{equation}\label{equ:J}
	J(\iparam) :=   \frac12  \sqwnorm{ \F\iparam - \obs }{ \Cobsnoise^{-1} }
	+ \frac12  \sqwnorm{ \iparam - \iparb }{ \Cparamprior^{-1}}.
\end{equation}
The Hessian of the above functional, is given by
\begin{equation}\label{eqn:Hessian}
	\Hessmat = \HMmat + \Cparampriormat^{-1} = \Cparampostmat^{-1},
\end{equation}
where $\HMmat = \F^* \Cobsnoise^{-1} \F$ denotes the Hessian of the data-misfit term in~\eqref{equ:J}.

Note that for the linear operator $\F$, we use $\F^*$ to denote its adjoint.
The reason we do not simply use matrix transpose is as follows.  The underlying
infinite-dimensional Bayesian inverse problem is formulated on $L^2(\D)$
equipped with the standard $L^2(\D)$ inner product. As noted
in~\cite{Bui-ThanhGhattasMartinEtAl13}, when discretizing the Bayesian inverse
problem, we also need to pay attention to the choice of inner products. Namely,
the discretized parameter space $\Rnum^\Nparam$ has to be equipped with a
discretized $L^2(\D)$ inner product. If finite element method is used to
discretize the inverse problem, then, the appropriate inner product to use is
the Euclidean inner product weighted by the finite-element mass matrix
$\mat{M}$. That is, for $\vec{u}$, $\vec{v}$ in $\Rnum^\Nparam$, we use
$\wip{\vec{u}}{\vec{v}}{\mat{M}} = \vec{u}\tran \mat{M} \vec{v}$.  
In the infinite-dimensional setting, the forward operator is
a mapping from $L^2(\D)$ to $\Rnum^\Nobs$; the domain is
endowed with the $L^2(\D)$ inner product and the co-domain is 
equipped with the Euclidean inner product, 
which we denote by $\ip{\cdot}{\cdot}$. Upon discretization, 
we work with the discretized forward operator
$\F:(\Rnum^{\Nparam},\wip{\cdot}{\cdot}{\mat{M}}) \to (\Rnum^{\Nobs},
\ip{\cdot}{\cdot})$. For $\vec{u} \in \Rnum^{\Nparam}$ and
$\vec{v} \in \Rnum^{\Nobs}$, we have
\[
    \ip{\mat{F}\vec{u}}{\vec{v}} = (\mat{F}\vec{u})\tran \vec{v}
    = \vec{u}\tran \mat{F}\tran \vec{v} = 
      \vec{u}\tran \mat{M} \mat{M}^{-1} \mat{F}\tran \vec{v}
    = \wip{\vec{u}}{\mat{M}^{-1}\mat{F}\tran\vec{v}}{\mat{M}};
\] 
from this we note
$\F^* = \mat{M}^{-1}\mat{F}\tran$. 
See~\cite{Bui-ThanhGhattasMartinEtAl13} for further details
on finite-element discretization of Bayesian inverse problems.

\subsection{Bayesian Optimal design of experiments}\label{subsec:OED}

Next, we turn to the problem of optimal design of experiments (ODE) for Bayesian
linear inverse problems governed by PDEs, which has received great amount of
attention in recent
years~\cite{bisetti2016optimal,HuanMarzouk13,HuanMarzouk14,LongScavinoTemponeEtAl13,long2015laplace}.
In a standard Bayesian experimental design problem, we seek an experimental design
that results in minimized \textit{posterior uncertainty} in the inferred
parameter $\iparam$. The way one chooses to quantify posterior
uncertainty leads to the choice of the design
criterion~\cite{AtkinsonDonev92,bernardo1979expected,box1982choice,ChalonerVerdinelli95,dette1996note,elfving1952optimum,KieferWolfowitz59,lindley1956measure,Ucinski05,Pukelsheim93}.
When the posterior distribution is Gaussian, the standard experimental design criteria are defined as functionals of
$\wCparampostmat$; here $\vec{w}$ denotes a generic vector of experimental design 
parameters.
For example, the A-optimal design is found by minimizing the trace of the posterior covariance operator
\begin{equation} \label{eqn:standard_A_optim_optimization_crit}
	\Psi^{\mathup{A}}(\vec w) := \trace{\left( \Cparampostmat(\design) \right)} \,.
\end{equation}
On the other hand, a D-optimal design is obtained by minimizing the log-determinant of the posterior covariance
\begin{equation} \label{eqn:standard_D_optim_optimization_crit}
	\Psi^{\mathup{D}}(\design ) := { \log{ \det {\left(  \Cparampostmat(\design)  \right)}} }
	= \log{ \det{\left( \left[ \Hessmat(\design) \right]^{-1} \right)} }  \,.
\end{equation}
See
e.g.,~\cite{alexanderian2016bayesian,AtkinsonDonev92,HaberHoreshTenorio08,HaberMagnantLuceroEtAl12}
for further details. 
As discussed in the introduction, the present work is focused on goal-oriented 
optimal design of experiments. This is detailed in the next section.

\section{Goal-Oriented Optimal Design of Experiments}\label{sec:GOODE}

Classical Bayesian optimal experimental design constructs experimental designs that result 
in minimized posterior
uncertainty on the inversion parameter $\iparam$. On the other hand, 
goal-oriented optimal design of experiments (GOODE) seeks
designs that minimize the uncertainty associated with 
a goal quantity of interest (QoI), 
which is a function of $\iparam$.
In this article, we consider a goal QoI of the form
\begin{equation}\label{eqn:GOOED_linear_QoI}
	\pred =  \Predmat \iparam,
\end{equation}
where $\Predmat$ is a linear operator, which we call
the \emph{goal operator}. 
The parameter $\iparam$ is inferred by solving
an inverse problem, as described in the previous section. 

As mentioned before, to keep the presentation simple, we work with discretized
quantities. The goal operator $\Predmat$ is thus the discretization of a linear
transformation that maps the inversion parameter, an element of $L^2(\D)$, 
to a goal QoI. We consider the case where the goal QoI is finite-dimensional;
i.e., the goal operator has a finite-dimensional range 
independent of discretization.
This is motivated by many applications in which the end goal 
is either a scalar or a relatively low-dimensional vector.
We denote the dimension of the goal $\pred$ by $\Npred$.

Assuming the Gaussian linear setting presented above, the prior and posterior laws of $\pred$
defined in~\eqref{eqn:GOOED_linear_QoI} can be obtained as follows: 
the prior
distribution law of $\pred$ is $\GM{\predb}{\Cpredpriormat}$, with
\begin{equation}\label{eqn:prior_prediction_PDF}
	\predb = \Predmat \iparamb,
	\qquad
	\Cpredpriormat = \Predmat \Cparampriormat \Predmat^*,
\end{equation}
where $\Predmat^*$ is the adjoint of the goal operator $\Predmat$.
The posterior distribution law of the goal QoI $\pred$, conditioned on the
observations $\obs$, is also Gaussian and is given by $\GM{\preda}{\Cpredpostmat}$, where
\begin{equation}\label{eqn:posterior_prediction}
	\begin{aligned}
		\preda = \Predmat \iparama,
		\qquad
		\Cpredpostmat = \Predmat \Cparampostmat \Predmat^*
		= \Predmat\, \left( \HMmat + \Cparampriormat^{-1} \right)^{-1} \, \Predmat^*.
	\end{aligned}
\end{equation}

Below, we describe the GOODE criteria under consideration, and present a
scalable framework for computing these criteria and their derivatives with
respect to design parameters, which exploits the fact that the 
end goal $\pred$ is of much lower dimension than $\iparam$. 

We focus on goal-oriented A- and D-optimal experimental designs (see
Section~\ref{sec:criteria}).  
We also examine these criteria from a decision theoretic point of view.  In the case of goal
oriented Bayesian A-optimality, we show that the minimization of expected Bayes
risk of the goal QoI is equivalent to minimizing the trace of the
covariance operator $\Cpredpostmat$. In the case of goal-oriented Bayesian
D-optimality, we show that maximizing the expected information gain for 
the goal QoI is equivalent to minimizing log-determinant of
$\Cpredpostmat$.  These results, which are natural extensions of the known
results from classical Bayesian optimal experimental design theory, provide
further insight on the interpretation of the presented GOODE criteria.
In Section~\ref{sec:sensors}, we describe the precise definition of an
experimental design vector that parameterizes a sensor placement and leads to
formulation of a suitable optimization problem. This is followed by the
discussion of the optimization problem for finding  A- and D-GOODE criteria and
a computational framework for computing the associated objective functions and
their gradients in Section~\ref{subsec:GOODE_optimization}. 
Algorithmic descriptions of the proposed methodologies, along with a discussion 
of the computational cost, are outlined in Section~\ref{subsec:GOODE_algorithms}.
In Section~\ref{subsec:ext}, we further discuss the connections of the GOODE 
criteria to the corresponding classical experimental design criteria, and outline 
possible extensions to the cases of nonlinear goal operators.


\subsection{GOODE criteria}\label{sec:criteria}

Let, as before, $\design$ denote a generic vector of experimental design parameters. 
(The precise definition of $\design$, in the case of optimal sensor placement problems, 
will be provided later in this section.)
In this section, we
will assume that $\Predmat$ has full row-rank, and 
consider a Bayesian linear inverse problem as formulated in
Section~\ref{subsec:Bayesian_inversion}.
%
%
Here we propose and examine the implications of goal-oriented A- and D-optimal criteria.

\paragraph{Goal-oriented A-optimal criterion }
%
The goal-oriented A-optimal design (A-GOODE) criterion is defined to be the trace 
of the posterior covariance matrix $\wCpredpostmat$ of the goal QoI: 
\begin{equation}
	\label{eqn:A_optim_optimization_crit} \Psi^\GA(\design) := \trace{\left(
	\wCpredpostmat \right)} = \trace(\Predmat \Cparampostmat(\design)
	\Predmat^*).  
\end{equation}
%
%
Two special cases are worth pointing out here. In the special case $\Predmat$
is a row vector $\Predmat = \vec{c}\tran$, we get $\Psi^\GA(\design) =
\vec{c}\tran{\Cparampostmat} \vec{c}$, which is the classical C-optimality
criterion. On the other hand, if $\Predmat = \mat{I}$, then this is nothing but
the classical Bayesian A-optimal criterion.

Note that while $\Cparampostmat(\vec{w})$ is a high-dimensional operator, in
practice usually $\Predmat$ has a low-dimensional range. Thus, in practice
$\wCpredpostmat$ is a low-dimensional operator. From computational point of
view, this means computing the A-GOODE criterion is significantly cheaper than
that of the classical A-optimality which is $ \trace(\Cparampostmat)$. 

Minimizing the A-GOODE criterion can be understood as
minimizing the average variance of the goal QoI.  
However, having a measure of the 
statistical quality of the estimator $\preda^\obs$ is also of crucial importance. 
Theorem~\ref{thm:bayesrisk} below addresses this by relating the
goal-oriented A-optimal criterion and the expected Bayes risk of $\preda^\obs$, which
is nothing but the mean squared error averaged over the prior
distribution. 


\begin{theorem}\label{thm:bayesrisk}
Consider a Bayesian linear inverse problem as formulated in 
Section~\ref{subsec:Bayesian_inversion}.
Let $\preda^\obs = \Predmat \iparama$ be the estimator for the 
goal QoI, $\vec\rho = \Predmat \iparam$, 
where $\Predmat$ has full row-rank, then,
	\begin{equation}\label{equ:bayesrisk2}
		\Expect{\priorm}{ \Expect{\obs |\iparam}{ \sqnorm{\preda^\obs - \pred}} }
		= \Trace{\Cpredpostmat}.
	\end{equation}
\end{theorem}
\begin{proof}
	See Appendix~\ref{apdx:proof_of_bayesrisk}.
\end{proof}
Note that for notational convenience, and since the result holds 
pointwise in $\design$, 
we have suppressed the dependence on $\design$ in the statement of the theorem.

\begin{remark}
In Theorem~\ref{thm:bayesrisk}, if we choose $\Predmat = \mat{I}$,
i.e., classical Bayesian A-optimal experimental design for a Gaussian posterior
distribution, we recover the known result~\cite{ChalonerVerdinelli95}
	\[
	\Expect{\iparam}{ \Expect{\obs |\iparam}{ \sqnorm{\iparama- \iparam}}}
	= \Trace{\Cparampostmat}.
	\]
	See also~\cite{AlexanderianGloorGhattas16} for a derivation of the above
	expression in the infinite-dimensional setting.
\end{remark}

\paragraph{Goal-oriented Bayesian D-Optimality criterion}
The goal-oriented D-optimal design (D-GOODE) criterion is taken to be the log-determinant of the posterior end-goal covariance
$\Psi^\GD$: 
%
\begin{equation} \label{eqn:D_optim_optimization_crit}
	\Psi^\GD(\design ) := { \log{ \det {\left(  \wCpredpostmat  \right)}} }
	= \log{ \det{\left( \Predmat \left[ \wHessmat \right]^{-1} \Predmat^* \right)} }. 
\end{equation}

To explain the motivation for this optimality criterion, 
we consider the Kullback-Leibler (KL) divergence~\cite{KullbackLeibler51}
from the posterior to prior distribution of the end-goal QoI 
$\pred$:
\begin{equation}\label{equ:KL}
	\begin{aligned}
		\DKL{\Pa(\pred|\obs,\design)}{\Pb(\pred)} &=  \DKL{ \GM{ \wpreda }{ \wCpredpostmat } }{  \GM{ \predb }{ \Cpredpriormat } }.
	\end{aligned}
\end{equation}
Since both distributions are Gaussian, the KL-divergence has a closed form
expression, which will simplify the calculations considerably.  The expected
information gain is defined as
\begin{equation}\label{eqn:Information_Gain}
	\Psi^{\rm KL} (\design) =
	\Expect{\priorm}{ \Expect{\obs|\iparam,\design}{ \DKL{\Pa(\pred|\obs,\design)}{\Pb(\pred)}  } }.
\end{equation}
The classical Bayesian D-optimality criterion is related to the expected information
gain, quantified by the expected  KL divergence
between the posterior distribution and the prior distribution. A  similar
relation for the goal-oriented D-optimality criterion $\Psi^\GD$ can be
derived. We present the following result that relates $\Psi^\GD$ and the
expected information gain $\Psi^\KL$:

\begin{theorem}\label{thm:expinfogain}
Consider a Bayesian linear inverse problem as formulated in
Section~\ref{subsec:Bayesian_inversion}.
Let $\pred = \Predmat \iparam$ be an end-goal QoI,
where $\Predmat$ has full row-rank.
Then,
\begin{equation}\label{equ:DKL}
\Expect{\priorm}{ \Expect{\obs|\iparam}{ \DKL{\Pa(\pred|\obs)}{\Pb(\pred)}  } }
= -\frac12 
 \log{ \det {\left(  \Cpredpostmat  \right)}} 
+ \frac12  \log\det\Cpredprior.
\end{equation}
\end{theorem}
\begin{proof} See Appendix~\ref{appen:On_D_Optimality}. \end{proof}
The significance of Theorem~\ref{thm:expinfogain} is that it says minimizing $\Psi^\GD$ with the
appropriate constraints on $\design$, amounts to maximizing the expected
information gain $\Psi^{\rm KL}$ under the same constraints.
\begin{remark}
Note that in the case of $\Predmat = \mat{I}$, and in the limit as $\Nparam \to \infty$
the criterion~\eqref{eqn:D_optim_optimization_crit} is meaningless. The reason for this
is that in the infinite-dimensional limit $\Cparampost$ is a positive self-adjoint trace class operator
and thus its eigenvalues accumulate at zero. On the other hand
in the case of $\Predmat = \mat{I}$,~\eqref{equ:DKL} simplifies to
\[
\Psi^{\rm KL}(\design) =
-\frac12 \log\det\Cparampost(\design) + \frac12 \log\det\Cparamprior = \frac12
\log\det(\Cparamprior^{1/2} \HMmat \Cparamprior^{1/2} + \mat{I}),
\]
which recovers the known result~\cite{AlexanderianGloorGhattas16} that for a
Gaussian linear Bayesian inverse problem
\begin{equation}\label{equ:DKLI}
	\Expect{\priorm}{ \Expect{\obs|\iparam,\design}{
	\DKL{ \GM{ \wiparama }{ \wCparampostmat } }{  \GM{ \iparprior }{ \Cparampriormat } }
	}}
	=
	\frac12
	\log\det(\Cparamprior^{1/2} \HMmat \Cparamprior^{1/2} + \mat{I}),
\end{equation}
and note that the quantity to the right is well defined in the limit as $\Nparam \to \infty$.
\end{remark}
The previous remark shows that~\eqref{equ:DKLI} is the correct expression for the
expected information gain to
choose in the case of $\Predmat = \mat{I}$.  While the derivation
of~\eqref{equ:DKLI} here is done for the discretized version of the problem, as
shown in~\cite{AlexanderianGloorGhattas16} such an expression for the expected
information gain can be derived in the infinite-dimensional Hilbert space
setting also. 

\subsection{Goal-oriented sensor placement}\label{sec:sensors} The experimental
conditions we choose to control are the sensor locations in the domain at which
data are to be collected. This can be expressed as an optimal design of
experiments (ODE) problem, as we now demonstrate.  Our strategy is to fix an
array of candidate locations for $\Nsens$ sensors and then select an optimal
subset of the candidate sensor locations.  In this context, a design $\design$
is a binary vector where each entry corresponds to whether or not a particular
sensor is active. Practical considerations, such as budgetary or physical
constraints, limit the number of sensors that can be chosen. In this context,
ODE seeks to identify the best possible sensor locations out of the possible
sensor locations.

We consider Bayesian inverse problems with time-dependent linear forward
models.  In this case, $\F$ maps the inversion parameter $\iparam$ to
spatio-temporal observations of the state at the sensor locations and at
observation times. In what follows, we make use of the notation $\F_{0,k}$ for
the forward model that maps $\iparam$ to the equivalent sensor measurements at
observation time instance $\tind{k}$.

In the present formulation, $\design$ enters the Bayesian inverse problem 
through the data likelihood:
\begin{equation}\label{eqn:weighted_joint_likelihood}
		\Prob(\obs | \iparam; \design) \propto
		\exp{\left( - \frac{1}{2} (\F(\iparam) - \obs)\tran\wdesignmat(\F(\iparam) - \obs)  \right) }
\end{equation}
where $\wdesignmat = \designmat^{1/2}\Cobsnoise^{-1} \designmat^{1/2}$, and
$\designmat \in \Rnum^{\Nobs \times \Nobs}$ is a block diagonal
matrix with $\Nobs = \Nsens \nobs$. 
In particular, $\designmat = \mathbf{I}_{\nobs} \otimes \designmat_s$
where $\designmat_s = \diag{w_1,\ldots,w_{\Nsens}}$ and $\otimes$ is the
Kronecker product.  
The noise covariance $\Cobsnoise$ is in general a $\Nobs \times \Nobs$ block diagonal
matrix $\Cobsnoise = \blkdiag{\mat{R}_1, \mat{R}_2, \cdots, \mat{R}_\nobs}$,
where $\mat{R}_k \in \Rnum^{\Nsens \times \Nsens}$ is the spatial noise covariance matrix corresponding to $k$th
observation time, and $\nobs$ is the number of observation time instances. 
In the present work, we assume that observations are uncorrelated in space and time,
and thus $\Cobsnoise$ is a diagonal
matrix. While this assumption is not necessary, it simplifies the formulation
considerably. Since $\designmat$ is also diagonal, we have the convenient
relation 
\[
   \wdesignmat = \designmat^{1/2}\Cobsnoise^{-1} \designmat^{1/2}  =
   \Cobsnoise^{-1/2} \designmat \Cobsnoise^{-1/2}.
\]

The posterior covariance of the parameter $\iparam$ is
\begin{equation}\label{eqn:weighed_prediction_posterior}
		\wCparampostmat = \left[ \wHessmat \right]^{-1}
		= \left( \wHMmat + \Cparampriormat^{-1} \right)^{-1}
		= \left( \F^* \wdesignmat \F + \Cparampriormat^{-1} \right)^{-1}.
\end{equation}
Therefore, the posterior distribution of the goal QoI $\pred$, conditioned by
the observations $\obs$, and the design $\design$
is the Gaussian $\GM{\wpreda}{\wCpredpostmat}$ with
\begin{equation}\label{eqn:weighted_posterior_prediction}
		\begin{aligned}
			\wpreda = \Predmat \wiparama \,,
			\qquad
			\wCpredpostmat = \Predmat \wCparampostmat \Predmat^*
			= \Predmat\, \left( \wHMmat + \Cparampriormat^{-1} \right)^{-1} \, \Predmat^*  \,.
		\end{aligned}
\end{equation}
The above definition of $\wCpredpostmat$ will be
substituted in the GOODE criteria described above.


Identifying the best $k$ sensor locations out of a set of $\Nsens$ candidate sensor 
locations is a combinatorial problem that is computationally intractable even
for modest values of $\Nsens$ and $k$. 
A standard
approach~\cite{papalambros2000principles,HaberHoreshTenorio08,HaberHoreshTenorio10,
HoreshHaberTenorio10,AlexanderianPetraStadlerEtAl14}, which we follow in this
article, is to relax the binary condition on the design weights and let $w_i \in
[0, 1]$, $i = 1, \ldots, \Nsens$. To ensure that only a limited number of
sensors are allowed to be active, we use sparsifying penalty functions to
control the sparsity of the optimal designs.
 
In the present formulation, non-binary weights are difficult to
interpret and implement. Thus, some form of thresholding scheme is required to make
a computed optimal design vector into a binary design.  In this work, we adopt the following
heuristic for thresholding: assuming that only $k$ sensors are to be placed in
the candidate $\Nsens$ locations, as is common practice, the locations
corresponding to highest $k$ weights can be selected. This means that the
corresponding weights are set to $1$ whereas all other sensors are set to $0$,
therefore giving a ``near-optimal'' solution to the original binary
problem~\cite{krause2008near}.  An alternative approach, not considered in this
article, is to obtain binary weights is to successively approximate the
$\ell_0$-norm by employing a sequence of penalty functions yielding a binary
solution~\cite{AlexanderianPetraStadlerEtAl14}.  The issue of sparsification
and the choice of sparsifying penalty function is elaborated further in the 
description of the GOODE optimization problem formulation below and in the
numerical results section. 

\subsection{The optimization problem} \label{subsec:GOODE_optimization}
The generic form of the goal-oriented experimental design problem involves an optimization problem
of the form
\begin{equation} \label{eqn:optim_optimization_vector}
	\begin{aligned}
		& \min_{\design \in \Rnum^{\Nsens}}{ \Psi (\design) + \alpha \, \Phi(\design) }  \\
		\text{subject to}\quad & 0 \leq w_i \leq 1, \quad i = 1, \ldots, {\Nsens} \,,
	\end{aligned}
\end{equation}
where, $\Psi$ is the specific design criterion, $\Phi(\design):
\Rnum_{+}^{\Nsens} \mapsto [0, \infty)$ is a penalty function, and $\alpha > 0$
is a user-defined penalty parameter that controls sparsity of the design.  In
this work, we make a choice to incorporate an $\ell_1$ norm to control sparsity
of the design; that is, we set 
\begin{equation} \Phi(\design):=\norm{\design}_1.
\end{equation} Depending on
whether we want goal-Oriented A- or D-optimality, $\Psi$ is either $\Psi^\GA$
or $\Psi^\GD$, respectively.

The optimization problem is solved using a gradient based approach, and therefore, this
requires the derivation of the gradient. Since the design weights are
restricted to the interval $[0,\,1]$, $\Phi(\design)$ differentiable, and the
gradient  of the penalty term in~\eqref{eqn:optim_optimization_vector} is
$\alpha \, \one$, where $\one \in \mathbb{R}^{\Nsens}$ is a vector of ones. In
the sequel, we derive expressions for the gradient of $\Psi^\GA$ and
$\Psi^\GD$. 
\subsubsection{Gradient of $\Psi^\GA$}
The gradient of $\Psi^\GA$ with respect to the design is given by
\begin{subequations} \label{eqn:A_optim_discrete_vector_gradient}
\begin{equation}
	\nabla_\design \Psi^\GA = - \sum_{k=1}^{\nobs}{
		\sum_{j=1}^{\Npred}{ \vec{\zeta}_{k,j} \odot \vec{\zeta}_{k,j} }
	},
\end{equation}
where  $\odot $ is the pointwise Hadamard product, and
	\begin{equation}\label{eqn:zetakj}
	\vec{\zeta}_{k,j} = \Cobsnoisemat_k^{-\frac{1}{2}} \F_{0,k}  \, \left[\wHessmat\right]^{-1} \Predmat^* \, \vec{e}_i \,.
\end{equation}
\end{subequations}
Here  $\vec{e}_i$ is the $i^{th}$ coordinate vector in
$\Rnum^{\Npred}$, and $\F_{0,k}$ is the forward model that maps the parameter
to the equivalent observation at time instance $\tind{k}$,
$k = 1,2,\ldots,\nobs$. See Appendix~\ref{appen:A_optim_discrete_vector_gradient}
for derivation of this gradient expression.

In the case the forward model is \emph{time-independent}, the expressions simplify to 
\begin{equation} \label{eqn:A_optim_discrete_vector_gradient2}
	\nabla_\design \Psi^\GA = - \sum_{i=1}^{\Npred}{
	\left( \Cobsnoise^{-\frac{1}{2}} \F  \, \left[\wHessmat\right]^{-1} \Predmat^* \, \vec{e}_i \right) \odot
	\left( \Cobsnoise^{-\frac{1}{2}} \F  \, \left[\wHessmat\right]^{-1} \Predmat^* \, \vec{e}_i \right)
	}
	\,,
\end{equation}
where $\vec{e}_i$ is the $i^{th}$ coordinate vector in $\Rnum^{\Npred}$.

\subsubsection{Gradient of $\Psi^\GD$} We present two alternate way of
computing the gradient that are equivalent, but differ in computational cost
depending on the number of sensors $\Nsens$ and 
dimension of the goal QoI, i.e. $\Npred$.
In the first formulation we assume that $\Nsens \geq \Npred$. We can compute the gradient as
\begin{subequations}\label{eqn:D_optim_discrete_vector_gradient_form_1_intxt}
	\begin{equation}\label{eqn:D_optim_discrete_vector_gradient_form_1}
		\nabla_{\design}\left( \Psi^\GD(\design) \right)
		= - \sum_{k=1}^{\nobs}{  \sum_{j=1}^{\Npred}{ \vec{\xi}_{k,j} \odot \vec{\xi}_{k,j} } }  \,,
	\end{equation}
	where
	\begin{equation}\label{eqn:xikj}
		\vec{\xi}_{k,j} = \Cobsnoisemat_k^{-1/2} \F_{0,k}  \,
		\left[\wHessmat\right]^{-1} \Predmat^*  \Cpredpostmat^{-1/2}(\design) \,
		\vec{e}_j,
	\end{equation}
\end{subequations}
and  $\vec{e}_j$ is the $i^{th}$ coordinate vector in
$\Rnum^{\Npred}$. The details are given in
Appendix~\ref{appen:D_optim_discrete_vector_gradient_form1}.  
This formulation is especially suited for large-scale four-dimensional variational (4D-Var) data assimilation applications~ \cite{attia2015hmcsampling,attia2015hmcsmoother,navon2009data}, including weather forecasting, and ocean simulations. 
Note that, in
this formulation, evaluating the gradient requires the square-root of the
posterior covariance matrix $\wCpredpostmat$.

Evaluating~\eqref{eqn:D_optim_discrete_vector_gradient_form_1_intxt}, requires solving $\Npred$ linear systems. 
If $\Nsens < \Npred$, the following alternative formulation
of the gradient of $\Psi^\GD$ will be computationally beneficial:
\begin{subequations}\label{eqn:D_optim_discrete_vector_gradient_form_2_intxt}
	\begin{equation}\label{eqn:D_optim_discrete_vector_gradient_form_2}
		%
                \nabla_{\design} \left( \Psi^\GD(\design ) \right)
                 = - \sum_{k=1}^{\nobs}{ \sum_{i=1}^{\Nsens} {\vec{e}_i \left( \vec{\eta}_{k,i}\tran \Cpredpostmat^{-1}
                   \vec{\eta}_{k,i} \right) } } 
	\end{equation}
where $\vec{\eta}_{k,i}$ is now:
	\begin{equation}\label{eqn:etaki}
		\vec{\eta}_{k,i} = \Predmat\, \left[ \wHessmat \right]^{-1} \F_{k,0}^*  \, \Cobsnoisemat_k^{-1/2}  \vec{e}_i       \,,
	\end{equation}
\end{subequations}
and $\vec{e}_i$ is the $i^{th}$ coordinate vector in $\Rnum^{\Nsens}$. 
The derivation details are given in Appendix~\ref{appen:D_optim_discrete_vector_gradient_form2}.
Note that the two gradient expressions are equivalent, in exact arithmetic.
	
In the \emph{time-independent} setting, where a single vector of sensor measurements is 
available, the gradient expressions for $\Psi^\GD$ simplify as follows.
The formula~\eqref{eqn:D_optim_discrete_vector_gradient_form_1_intxt} reduces to
\begin{subequations}\label{eqn:time_indep_D_grad}
	\begin{equation}
		\nabla_{\design}\left( \Psi^\GD(\design) \right)
		= - \sum_{j=1}^{\Npred}{  \vec{\xi}_{j} \odot \vec{\xi}_{j} }  \, \quad \text{with }
		\quad \vec{\xi}_{j} = \left( \Cobsnoise^{-1/2} \F  \, \left[\wHessmat\right]^{-1} \Predmat^*  \Cpredpostmat^{-1/2}(\design) \, \vec{e}_j \right) \,,
	\end{equation}
	where $\vec{e}_j$ is the $j^{th}$ coordinate vector in $\Rnum^{\Npred}$, 
        and~\eqref{eqn:D_optim_discrete_vector_gradient_form_2_intxt} reduces to
	\begin{equation}
	    \nabla_{\design}\left( \Psi^\GD(\design) \right)
            = - { \sum_{i=1}^{\Nsens} {\vec{e}_i \left( \vec{\eta}_{i}\tran \Cpredpostmat^{-1}
                \vec{\eta}_{i} \right) } }   \, \quad \text{with }
            \quad \vec{\eta}_{i} = \Predmat\, \left[ \wHessmat \right]^{-1} \F^*  \, \Cobsnoise^{-1/2}  \vec{e}_i \,,
        \end{equation}
        where $\vec{e}_i$ is the $i^{th}$ coordinate vector in $\Rnum^{\Nobs}$.
	\end{subequations}

\subsection{Implementation and computational considerations}\label{subsec:GOODE_algorithms}
%
The main bottleneck of solving A-GOODE and D-GOODE
problems lie in the evaluation of the respective objective functions and the associated gradients.
Here, we detail the steps of evaluating the objective function and the
gradient of both A-GOODE and D-GOODE problems.  These steps are explained in
Algorithms~\ref{alg:A_GOODE}, and~\ref{alg:D_GOODE}, respectively.  For
simplicity, we drop the penalty terms in the computations presented in
the two Algorithms~\ref{alg:A_GOODE}, and~\ref{alg:D_GOODE}.

\textbf{The A-GOODE problem}.
Algorithm~\ref{alg:A_GOODE} outlines the main steps of evaluating the
A-GOODE objective
and the
gradient~\eqref{eqn:A_optim_discrete_vector_gradient}.  Notice that all loops
over the end-goal dimension, in Algorithm~\ref{alg:A_GOODE}, i.e., steps
1--3, 8--10, and 16--20, are embarrassingly parallel.  Evaluating $\Predmat^*
\vec{e}_j$, $j=1,\ldots, \Npred$, is independent of the design $\design$, and
can be evaluated offline and saved for later use.  Step 2 requires a Hessian
solve for each of the vectors $\{g_j\}_{j=1,\ldots,\Npred}$, which can be done
using preconditioned conjugate gradient.  Each application of the Hessian
requires a forward and an adjoint PDE solve.  Using the prior covariance as a
preconditioner, this requires $\mathcal{O}(r)$ CG iterations, i.e.,
$\mathcal{O}(2 r)$ PDE solves, where $r$ is the numerical rank of the
prior-preconditioned data misfit Hessian;
see~\cite{bui2013computational,IsaacPetraStadlerEtAl15}.  
Another $\Npred$ forward solutions of the underlying PDE are required in Step
$17$ in the algorithm, for gradient computation.  Therefore, the total number
of PDE solves for objective and gradient evaluation is $\mathcal{O}(\Npred r)$.
This computations can be easily parallelized in $\Npred$ cores.
Moreover, the applications of the inverse
Hessian---the Hessian solves---can be accelerated by using low-rank
approximations of the prior-preconditioned data misfit Hessian~\cite{bui2013computational}.
Note that Algorithm~\eqref{alg:A_GOODE} also requires 
$\Npred$ independent applications of $\Predmat$ and its adjoint, which
can be done in parallel. 

		\begin{algorithm}[htbp!]  \scriptsize
			\renewcommand{\algorithmicrequire}{\textbf{Input:}}
			\renewcommand{\algorithmicensure}{\textbf{Output:}}
			\begin{algorithmic}[1]
				\Require 
                                $\design,\, \Predmat,\, \Hessmat,\, \Npred,\, \nobs,\, {\{\Cobsnoisemat_k\}}_{k=1,\ldots,\nobs},\, { \{ \F_{k-1,\, k} \} }_{k=1,\ldots,\nobs}$  

				\Ensure objective, grad
				\For{ $j = 1,\ldots, \Npred$}
				where $\vec{e}_j$ is the $j^{th}$ coordinate vector in $\Rnum^{\Npred}$
				\State solve $\left[ \Hessmat(\design) \right] \,
        \vec{g}_j = \Predmat^* \vec{e}_j$, for $\vec{g}_j$       \Comment{$\Predmat^*$ is the adjoint of the goal operator $\Predmat$}
				\EndFor
				\State objective $\leftarrow$
				\Call{A\_GOODE\_Objective}{$\design,\, \{\vec{g}_j \}_{j=1, \ldots, \Npred}$}
				\State grad $\leftarrow$
				\Call{A\_GOODE\_Grad}{$\design,\,  \{\vec{g}_j \}_{j=1, \ldots, \Npred}$}


				%
				\Function{A\_GOODE\_Objective}{ $\design,\, \{\vec{g}_j\}_{j=1,\ldots,\Npred}$}
				\State initialize objective = 0
				\For{ $j = 1,\ldots, \Npred$}
				\State objective $\leftarrow$ objective + $\vec{e}^T_j \Predmat \vec{g}_j$   \Comment{the objective function~\eqref{eqn:A_optim_optimization_crit}}
				\EndFor
				\State \Return objective  
				\EndFunction
				\Function{A\_GOODE\_Grad}{$\design,\,  \{\vec{g}_j \}_{j=1, \ldots, \Npred}$ }
				\State initialize grad = $\vec{0} \in \Rnum^{\Npred}$
				\For{ $k = 1,\ldots, \nobs$}
				\For{ $j = 1,\ldots, \Npred$ }
				\State update $\vec{g}_j \leftarrow \F_{k-1,\, k}\, \vec{g}_j $
				\State $\vec{\zeta}_{k,j} \leftarrow \Cobsnoisemat_k^{-1/2}\, \vec{g}_j$
				\State grad $\leftarrow$ grad - $\vec{\zeta}_{k,j} \odot \vec{\zeta}_{k,j}$  \Comment{the gradient~\eqref{eqn:A_optim_discrete_vector_gradient}}
				\EndFor
				\EndFor
				\State \Return grad
				\EndFunction
		\end{algorithmic}
	\caption{A-GOODE objective and gradient computation}
	\label{alg:A_GOODE}
	\end{algorithm}
	  \begin{remark}
      In Algorithm~\ref{alg:A_GOODE}, $\Predmat$ is the goal operator, $\Npred$
      is the dimension of the end-goal, $\nobs$ is the number of observation time
instances, $\design$ is the experimental design, $\Hessmat(\design)$ is the
weighted Hessian, $\Cobsnoisemat_k$ is the covariance of the measurement noise
at time instance $\tind{k}$, and $\F_{k-1,\, k}$ is the forward model that maps
the parameter from time instance $\tind{k-1}$ to the equivalent observation at
observation time instance $\tind{k}$. 
\end{remark}

	%
	%

\textbf{The D-GOODE problem.}
Algorithm~\ref{alg:D_GOODE} describes the main steps for
calculating D-GOODE objective and 
gradient expressions
\eqref{eqn:D_optim_discrete_vector_gradient_form_1_intxt}--\eqref{eqn:D_optim_discrete_vector_gradient_form_2_intxt}.  Similar to
Algorithm~\ref{alg:A_GOODE}, all loops over the end-goal dimension, 
i.e., steps 1--3, 14--22, and 16--20 in Algorithm~\ref{alg:D_GOODE}, are
inherently parallel. Also, the loop over the observation dimension,
$i=1,\ldots, \Nsens$, in Algorithm~\ref{alg:D_GOODE} is inherently parallel.  
With small end-goal space, the cost of Cholesky factorization of
$\Cpredpostmat$, in step $5$ is negligible compared to PDE solutions.  
The first form of the gradient, i.e.  steps 12--24, requires $\Npred$ Hessian
solves, and $\Npred$ forward PDE solves that can run completely in parallel.
The second form of the gradient, i.e. steps 25--41, on the other hand requires
$\Nsens$  Hessian solves, and $\Nsens$ forward and adjoint solves of the
underlying system of PDEs.  As mentioned before, low-rank approximation
of the prior-preconditioned data misfit Hessian can be used to accelerate
computations.  Checkpointing is utilized in the second form of the gradient,
i.e., steps 28--31.  The checkpointed solutions $\vec{r}_{k,i}$ are recalled in the
adjoint solves in step 33.  

Algorithm~\ref{alg:D_GOODE} requires $\Npred$ applications of $\Predmat$ and
its adjoint for objective function evaluations. As for the gradient, we need
$\Npred$ applications of $\Predmat^*$ with the first form of the gradient, and
$\nobs \Nsens$  applications of $\Predmat$. As before, 
the loops over end-goal and observation dimensions are embarrassingly 
parallel. 


\begin{algorithm}[htpb!] \scriptsize
	\renewcommand{\algorithmicrequire}{\textbf{Input:}}
	\renewcommand{\algorithmicensure}{\textbf{Output:}}
	\begin{algorithmic}[1]
	\Require $\design,\, \Predmat,\, \Hessmat,\, \Npred,\, \nobs,\, \Nsens,\, {\{\Cobsnoisemat_k\}}_{k=1,\ldots,\nobs},\, { \{ \F_{k-1,\, k} \} }_{k=1,\ldots,\nobs}$  

\Comment{
In addition to the arguments in
Algorithm~\ref{alg:A_GOODE}, here $\Nsens$ is
the number candidate sensor locations.} 
	\Ensure objective, grad
		\For{ $j = 1,\ldots, \Npred$}
			where $\vec{e}_j$ is the $j^{th}$ coordinate vector in $\Rnum^{\Npred}$
			\State solve $\left[ \Hessmat(\design) \right] \,
			\vec{g}_j = \Predmat^* \vec{e}_j$, for $\vec{g}_j$
			\State $ \vec{d}_j \leftarrow \, \Predmat \vec{g}_j$
		\EndFor
		\State calculate $\mat{L}$, the Cholesky factorization of $\Cpredpostmat $
		\State objective $\leftarrow\, 2\, \sum_{j=1}^{\Npred}{ \log{\left( {\left[ \mat{L}\right]_{jj}} \right)} }$    \Comment{the objective function~\eqref{eqn:D_optim_optimization_crit}}

		\If {$\Nsens<\Npred$}
			\State grad $\leftarrow$ \Call{D\_GOODE\_Grad\_1}{$\design,\, \mat{L}$}
		\Else
			\State grad $\leftarrow$ \Call{D\_GOODE\_Grad\_2}{$\design,\, \Cpredpostmat$}
		\EndIf

		\vspace{0.25cm}
		\SubAlg{Two forms of the function \textrm{D\_GOODE\_Grad} }
		\Function{D\_GOODE\_Grad\_1}{$ \design$,  $\mat{L}$ }  \Comment{First form~\eqref{eqn:D_optim_discrete_vector_gradient_form_1_intxt}}
			\State initialize grad = $\vec{0} \in \Rnum^{\Npred}$
			\For{ $j = 1,\ldots, \Npred$ }
				\State solve $\mat{L} \,\vec{r}_j = \vec{e}_j$, for $\vec{r}_j$
				\State solve $\left[ \Hessmat(\design) \right] \, \vec{q}_j = \Predmat^* \vec{r}_j$, for $\vec{q}_j$
				\For{ $k = 1,\ldots, \nobs$ }
				\State update $\vec{q}_j \leftarrow \F_{k-1,\, k}\, \vec{q}_j $
				\State $\vec{\xi}_{k,j} \leftarrow \Cobsnoisemat_k^{-1/2}\, \vec{q}_j$
				\State grad $\leftarrow$ grad - $\vec{\xi}_{k,j} \odot \vec{\xi}_{k,j} $
				\EndFor
				\EndFor
				\State \Return grad
			\EndFunction
				\vspace{0.15cm}
		\Function{D\_GOODE\_Grad\_2}{$ \design$,  $\Cpredpostmat$ }  \Comment{Second form~\eqref{eqn:D_optim_discrete_vector_gradient_form_2_intxt}}
				\State initialize grad = $\vec{0} \in \Rnum^{\Npred}$
				\For{ $i = 1,\ldots, \Nsens$ }
					\State calculate $\vec{r}_{1, i} \leftarrow \Cobsnoisemat_1^{-1/2} \vec{e}_i$
					\For{ $k = 2,\ldots, \nobs$ }
					\State calculate $\vec{r}_{k, i} \leftarrow \F_{k-1, k}\, \vec{r}_{k-1, i}$
						\hspace*{6em}%
								\rlap{\smash{$\left.\begin{array}{@{}c@{}}\\{}\\{}\\{}\end{array}\right\}%
														\begin{tabular}{l}  \end{tabular}$}}  \Comment{checkpointing}
					\EndFor
					\For{ $k = \nobs, \nobs-1, \ldots, 1$ }
						\State calculate $\vec{q}_{k,i} \leftarrow \F^*_{k,\, k-1}\, \vec{r}_{k,i}$   \Comment{$\F^*_{k,\, k-1}$ is the adjoint of the forward operator $\F_{k-1,\, k}$}
						\State solve $\left[ \Hessmat(\design) \right] \, \vec{\eta}_{k,i} = \vec{q}_{k, i}$, for $\vec{\eta}_{k,i}$
						\State update $\vec{\eta}_{k,i} \leftarrow \Predmat\, \vec{\eta}_{k,i}$
						\State solve $\Cpredpostmat \, \vec{\nu}_{k,i} = \vec{\eta}_{k,i}$, for $\vec{\nu}_{k,i}$
						\State grad $\leftarrow$ grad - $\left( \vec{\eta}_{k,i}^T \, \vec{\nu}_{k,i} \right)\, \vec{e}_i$
					\EndFor
				\EndFor
				\State \Return grad
		\EndFunction
		\end{algorithmic}
		\caption{D-GOODE objective and gradient computation}
		\label{alg:D_GOODE}
\end{algorithm}


\subsection{Connections to classical ODE criteria and extensions}\label{subsec:ext}
Here we further discuss the connections of the GOODE criteria to the
corresponding classical experimental design criteria and an extension to 
nonlinear goal operators.  
We have already mentioned two important connections:
\begin{enumerate}
\item Taking $\Predmat = \mat{I}$ trivially recovers the Bayesian A- and D-optimal
design criteria.  
\item In the special case $\Predmat$ is a row vector
$\Predmat = \vec{c}\tran$, we get $\Psi^\GA(\design) =
\vec{c}\tran{\Cparampostmat} \vec{c}$, which is the Bayesian C-optimality
criterion.
\end{enumerate}
Furthermore, if the vector $\vec{c}$ is randomly drawn from a distribution
$\pi$ with mean zero and identity covariance, then $\Expect{\pi}{\vec{c}\tran \Cparampostmat \vec{c}} 
= \trace(\Cparampostmat)$; that is, in
expectation, it is nothing but the classical Bayesian A-optimality. Thus, if
$\vec{c}$ is a single draw from $\pi$, then the scalar GOODE is an unbiased
trace estimator~\cite{AvronToledo11}.

Next, we discuss possible extensions to the case of nonlinear
end-goal operators. 
Let $\pred = \vec{p}(\iparam)$ denote a
nonlinear parameter-to-goal map. Assuming $\vec{p}(\iparam)$ is a
differentiable function of the parameter $\iparam$, one can consider a
linearization: 
\[
\vec{p}(\iparam)
\approx \vec{p}(\iparam_0) + \Predmat(\iparam_0)(\iparam - \iparam_0),
\]
where $\iparam_0$ is a reference (nominal) parameter value, and
$\Predmat(\iparam_0)$ is the Fr\'echet derivative of $\vec{p}(\iparam)$,
evaluated at $\iparam_0$. The linearization point $\iparam_0$ can be chosen,
for instance, by taking the mean of the prior distribution $\iparprior$.  An
alternative choice for the linearization point is the MAP-estimator
$\iparpost(\design)$. This has the advantage of incorporating the inverse problem
solution in the GOODE problem, but presents an added challenge: 
%
the data $\obs$ needed to compute $\iparpost(\design)$ is unavailable \textit{a
priori}. 
An approach to tackle this issue has been described
in~\cite{AlexanderianPetraStadlerEtAl16}, in the context of A-optimal design of
experiments for nonlinear inverse problems; the development of a similar
strategy for GOODE problems with nonlinear goal operators is subject
of future work.  


	\section{Model problem and experimental setup} \label{sec:Experiment_Setup}
	In this section, we detail a model Bayesian inverse problem, which we
	use to illustrate the criteria and the algorithms proposed in the
	present work.  The model problem is taken to be a contaminant source
	identification problem in which the spatio-temporal measurements of the
	contaminant field at sensor locations are used to estimate the source,
	or initial conditions, of the contaminant
	field~\cite{AkcelikBirosDraganescuEtAl05,FlathWilcoxAkcelikEtAl11,PetraStadler11}. 

	\paragraph{The forward operator}
	The governing equation of the contaminant field $u = \xcont(\mathbf{x}, t)$ is assumed to be the following advection-diffusion equation with associated boundary conditions:
	\begin{equation}\label{eqn:advection_diffusion}
		\begin{aligned}
			\xcont_t - \kappa \Delta \xcont + \vec{v} \cdot \nabla \xcont &= 0     \quad \text{in } \domain \times [0,T],   \\
			\xcont(0,\,x) &= \theta \quad \text{in } \domain,               \\
			\kappa \nabla \xcont \cdot \vec{n} &= 0     \quad \text{on } \partial \domain \times [0,T],
		\end{aligned}
\end{equation}
where $\kappa>0$ is the diffusivity, $T$ is the final time, and
$\vec{v}$ is the velocity field. This models the transport of the contaminant
field in the domain.  The domain $\mathcal{D}$, which is depicted in
Figure~\ref{fig:ad_diff}~(left), is the region $(0, 1) \times (0, 1)$ with the
shown rectangular regions in its interior excluded. 
These regions model buildings, in which the
contaminant does not enter. The boundary $\partial \domain$ includes both the
external boundary and the building walls.  The velocity field $\vec{v}$, shown
in Figure~\ref{fig:ad_diff}~(right), is obtained by solving a steady
Navier-Stokes equation as detailed in~\cite{PetraStadler11,
VillaPetraGhattas2016}.
	
\begin{figure}\centering
  \includegraphics[width=0.45\linewidth]{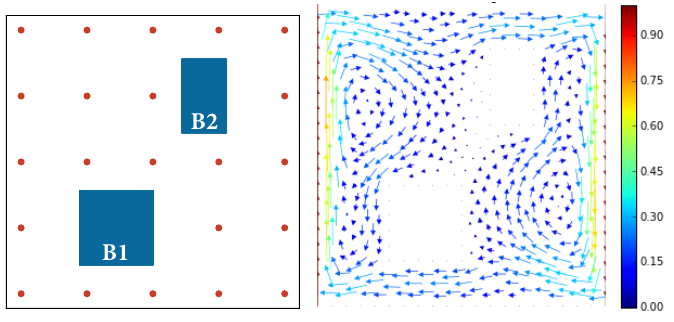}
\caption{Advection-diffusion model settings.  Left: The physical domain; the
red dots indicate candidate sensor locations.  Right:  The velocity field.}
\label{fig:ad_diff} \end{figure}

To evaluate the forward operator,
we solve~\eqref{eqn:advection_diffusion}, and then apply
a restriction operator (observation operator) $\mathcal{B}$ to the solution $u(\vec{x}, t)$ to extract
solution values at a set of predefined (sensor) locations $\{ \vec{x}_1,\,
\vec{x}_2,\, \ldots,\,  \vec{x}_{\nobs} \} \subset \domain$, at fixed time
instances $\{ t_1,\,  t_2,\, \ldots,\,   t_{\nobs} \} \subset [0, T]$.   The
parameter-to-observable map, i.e., the forward operator $\mathcal{F}
:L^2(\domain) \rightarrow \mathbb{R}^{\nobs}$ maps the initial condition to the sensor measurements.

We discretize the PDE using
Lagrange triangular elements of order $2$ with $\Nstate=7863$ spatial
degrees of freedom in space, and using implicit Euler in time.
We use an implementation of the forward model provided by
the inverse problems package HippyLib~\cite{VillaPetraGhattas2016}, which is
built on top of FEniCS~\cite{logg2012automated}. As before, the discretized forward
operator is denoted by $\F$.

\paragraph{Prior distribution and data likelihood} The prior
  distribution of the parameter $\iparam$ is Gaussian
  $\GM{\iparb}{\Cparampriormat}$, with $\Cparampriormat$ taken to be the
  discretization of $\mathcal{A}^{-2}$, where $\mathcal{A}$ is an Laplacian
  operator; see \cite{VillaPetraGhattas2016} for details.  The observation error
  covariance $\Cobsnoise$ is a diagonal matrix with variances calculated based on
  a noise level of $1\%$.
  Specifically, the standard deviation of measurement noise is
  calculated by multiplying the noise level by the maximum
  magnitude of the pointwise observations resulting
  from applying the forward operator $\F$ to a reference initial condition.  The
  synthetic observations used are created by adding white noise, with uncertainty
  level $1\%$. While this is a relatively small uncertainty level, we noticed
  that increasing the observation error variance doesn't affect the main
  conclusions of the article.

\paragraph{Experimental configuration and design of experiments} 
	A uniformly distributed observational grid is deployed in the domain
	(Figure~\ref{fig:ad_diff}~(left)), and the number of candidate sensor
	locations is taken to be $\Nsens = 22$.
	Furthermore, the initial and final simulation time is taken to be $t_0=0,$ and $T=t_F=0.8$, and observations are taken at time instances
	$\{t_k\}=\{0.4,\,0.6,\, 0.8\}$ respectively.
	The inverse problem is to infer the initial state, $\theta$,
	using measurements $\obs$ taken after the contaminant has been
	subjected to diffusive transport. To explore the potential of goal-oriented designs, we use several
  end-goal operators (summarized in Table~\ref{table:prediction_operators}). 
  These operators predict QoIs evaluated by solving the forward problem in a time-interval that goes
  beyond that of the inverse problem, and thus are considered \emph{prediction operators}.
	\begin{table}[!ht]
		\centering
		\resizebox{0.95\columnwidth}{!}{%
		\begin{tabular}{|l | p{6cm} | p{6cm}|}
			\cline{2-3}
			\multicolumn{1}{c|}{}  & \multicolumn{1}{c|}{\bf Vector-valued prediction}  & \multicolumn{1}{c|}{\bf Scalar-valued prediction} \\ \cline{2-3}
			\multicolumn{1}{c|}{}  & concentration of the contaminant observed within distance $\epsilon$ from the internal boundaries at time $\tind{pred}$& the \textit{``average''} concentration of the contaminant within distance $\epsilon$ from the internal boundaries at time $\tind{pred}$ \\ \hline
			{\bf First building walls} ($\mathbf{B1}$)   & \multicolumn{1}{c|}{$\Predmat_{\rm v0}$ }  & \multicolumn{1}{c|}{$\Predmat_{\rm s0} \equiv \vec{v}\tran \Predmat_{\rm v0}$ }  \\
			{\bf Second building  walls} ($\mathbf{B2}$) & \multicolumn{1}{c|}{$\Predmat_{\rm v1}$ }  & \multicolumn{1}{c|}{$\Predmat_{\rm s1} \equiv \vec{v}\tran \Predmat_{\rm v1}$ }  \\
			{\bf Both buildings walls } (combined)          & \multicolumn{1}{c|}{$\Predmat_{\rm v2}$ }  & \multicolumn{1}{c|}{$\Predmat_{\rm s2} \equiv \vec{v}\tran \Predmat_{\rm v2}$ }  \\ \hline
			\end{tabular}%
			}
      \caption{A summary of the goal operators used in
			the numerical experiments.  The vector-valued prediction carries
			out two steps; a) a forward propagation to the
			prediction time using model dynamics, b) a restriction
			operator that extracts the value of $\xcont$ at the
			prediction grid-points. The scalar-valued prediction
			averages the vector-valued prediction QoI. Here
			$\vec{v}$ is a vector with entries equal to $1 /
			\Npred$, where $\Npred$ is the dimension of the
			corresponding vector-valued prediction.  }

		\label{table:prediction_operators}
	\end{table}

	Specifically, we consider two general cases: vector-valued, and scalar-valued
	prediction QoIs.  A vector-valued prediction operator maps
	the initial parameter $\iparam$ to a prediction time $\tind{pred} > t_F$ using
	the forward model, then a restriction operator is applied to extract/observe
	the state at pre-specified prediction grid points.  The prediction grid-points
	are the grid points within a specific distance $\epsilon$ from the internal
	boundaries, i.e., walls of the buildings.  We use $\Predmat_{\rm v0},\,
	\Predmat_{\rm v1},\, \Predmat_{\rm v2}$ to denote vector-valued prediction
	operators within a distance $\epsilon$ from the boundary of first building $\mathbf{B1}$,
	second building  $\mathbf{B2}$, and the two buildings combined, respectively, as shown in Figure~\ref{fig:ad_diff}. 
	The scalar-valued prediction operators compute the average of
	the predictions arising out of the corresponding prediction operators; that is,
	\begin{equation}
		\Predmat_{\rm si} := \vec{v}\tran \Predmat_{\rm vi} , \quad i = 0,\,1,\,2\,,
	\end{equation}
	where $\vec{v} = 1/\Npred \begin{bmatrix} 1 & 1 & \cdots & 1\end{bmatrix}\tran \in \Rnum^{\Npred}$.
	Table~\ref{table:prediction_dimensions} summarizes, the prediction time
	$\tind{pred}$, the prediction distance $\epsilon$ from the corresponding
	internal boundary, and the dimension of the range of all vector-valued
	prediction operators tested herein.
	\begin{table}[]
		\centering
		\resizebox{0.35\linewidth}{!}{%
		\begin{tabular}{|l | p{1cm} | p{1cm} | p{1cm}|}
			\hline
			\multicolumn{1}{|c|}{\bf Prediction operator}  & \multicolumn{1}{c|}{$\tind{pred}$}  & \multicolumn{1}{c|}{$\epsilon$}  & \multicolumn{1}{c|}{$\Npred$} \\ \hline
			\multicolumn{1}{|c|}{$\Predmat_{\rm v0}$ }   & \multicolumn{1}{c|}{$1.0$}  & \multicolumn{1}{c|}{$0.02$}  &  \multicolumn{1}{c|}{$164$} \\
			\multicolumn{1}{|c|}{$\Predmat_{\rm v1}$ }   & \multicolumn{1}{c|}{$1.0$}  & \multicolumn{1}{c|}{$0.02$}  &  \multicolumn{1}{c|}{$138$}  \\
			\multicolumn{1}{|c|}{$\Predmat_{\rm v2}$ }   & \multicolumn{1}{c|}{$1.0$}  & \multicolumn{1}{c|}{$0.02$}  &  \multicolumn{1}{c|}{$302$} \\ \hline
			\end{tabular}%
			}
		\caption{Details of the prediction vector-valued operators used in the numerical experiments, and described in Table~\ref{table:prediction_operators}.
		The prediction time $\tind{pred} $, the prediction distance $\epsilon$ from the corresponding internal boundary, and the dimension of the range of all vector prediction operators used are shown.
		}
		\label{table:prediction_dimensions}
	\end{table}

	\paragraph{Sparsification strategy and optimization solver}
	As mentioned in Section \ref{subsec:GOODE_optimization}, we use an $\ell_1$-norm penalty to control the
	sparsity of the design. The design penalty parameter
	$\alpha$ is tuned empirically, as discussed in the numerical results below.
	The optimization problem~\eqref{eqn:optim_optimization_vector} is solved using a
	limited-memory quasi-Newton algorithm for bounded constrained optimization~\cite{zhu1997algorithm}.
	The optimization algorithm approximates second	derivatives using the Broyden-Fletcher-Goldfarb-Shanno (BFGS) method; see
	e.g.,~\cite{byrd1995limited,morales2011remark}.

	\section{Numerical Results}\label{sec:Numerical_Results}
	This section summarizes the numerical experiments carried out using the
	settings described in Section~\ref{sec:Experiment_Setup}.  In
	Section~\ref{subsec:scalar_goal_results}, we provide a numerical study of the
	GOODE problem with scalar-valued prediction QoI.
	Section~\ref{subsec:vector_goal_results} provides numerical experiments of the
	GOODE problem with vector-valued prediction QoI.  In Section~\ref{subsec:reg_parameter},
	we present an empirical study on the choice of the penalty parameter
	$\alpha$.

	\subsection{Scalar-valued goal}\label{subsec:scalar_goal_results}
	%
	This section contains numerical experiments of the GOODE problem with scalar prediction QoI, defined as the average contaminant within distance $\epsilon$ (see Table~\ref{table:prediction_dimensions}) from the building(s) boundaries. Note that for a scalar QoI, the A- and D-GOODE criteria are identical.

Figure~\ref{fig:scalar_GOOED_A_Optimal_weights_3d} shows the GOODE optimal
design, for several choices of the penalty parameter $\alpha$.  The optimal
weights $\{w_i\}_{i=1,\ldots, \Nsens}$ of the $\Nsens=22$ candidate sensor
locations are plotted on the z-axis, where the weights are normalized to sum
to $1$, after solving the optimization problem~\eqref{eqn:optim_optimization_vector}.
	\begin{figure}
		\centering
    \includegraphics[width=0.73\linewidth]{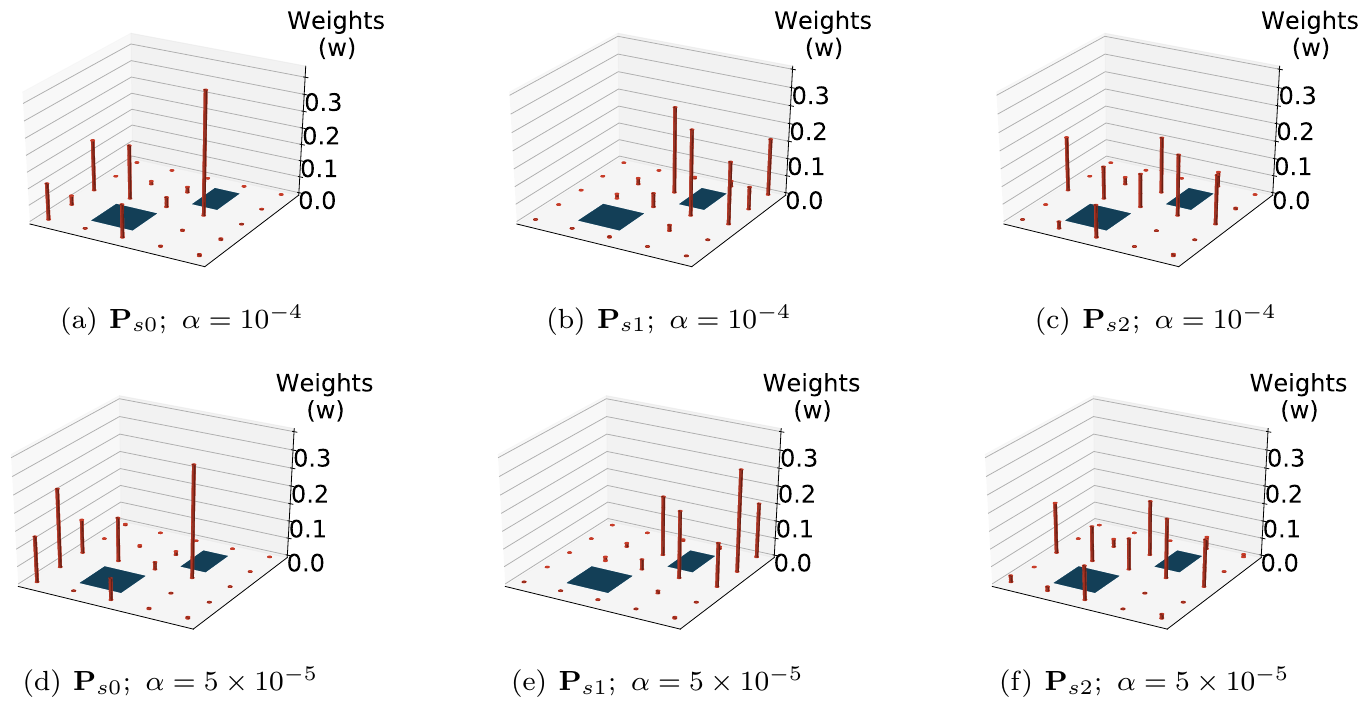}
		\caption{Scalar GOODE results for the model problem (Section~\ref{sec:Experiment_Setup}). The prediction operators are $\Predmat_{\rm si} $ for $i=1,2,3$.
		The optimal weights $\{w_i\}_{i=1,\ldots, \Nsens}$ are plotted on the z-axis, where the weights are normalized to add up to $1$.
		}
	\label{fig:scalar_GOOED_A_Optimal_weights_3d}
	\end{figure}
	%
	%

The results in Figure~\ref{fig:scalar_GOOED_A_Optimal_weights_3d},
show the utility of incorporating the end-goal in the solution of the
sensor placement problem.
Specifically, since the goal defined by $\Predmat_{\rm s0}$ is to predict the
average concentration of the contaminant at time $t_\pred$ around the first building,
the goal-oriented optimal design
(Figures 3(a), 3(d)) 
is a sparse solution with
highest weights concentrated around the first building where the prediction
information is maximized.  Similarly, the optimal relaxed design for solving
the GOODE problem with prediction operator $\Predmat_{\rm s1}$, is a design
with high weights centered around the second building 
(see Figures 3(b), 3(e)). 
In the last case where the
prediction operator $\Predmat_{\rm s2}$ is used, the goal is to predict the
average concentration of the contaminant around the two buildings combined.
We found that the designs for the goal that included both buildings, had a strong overlap with the designs obtained by considering the goals (individual buildings) separately.

As mentioned in Section~\ref{sec:sensors}, the sensors corresponding to the $k$
largest weights are set to $1$, the other sensor weights are set to zero. We
now present the thresholded solution Figure~\ref{fig:scalar_GOOED_A_Optimal_weights_upper_quartile}, obtained by thresholding the solution in Figure~\ref{fig:scalar_GOOED_A_Optimal_weights_3d}, with $k=5$. The number of sensors $k=5$ is deliberately chosen to be a small number, to enable comparison with  the
``brute-force'' solution found by enumerating all possible combinations of
active sensors; see Section~\ref{subsec:reg_parameter}.

	\begin{figure}
		\centering
    \includegraphics[width=0.62\linewidth]{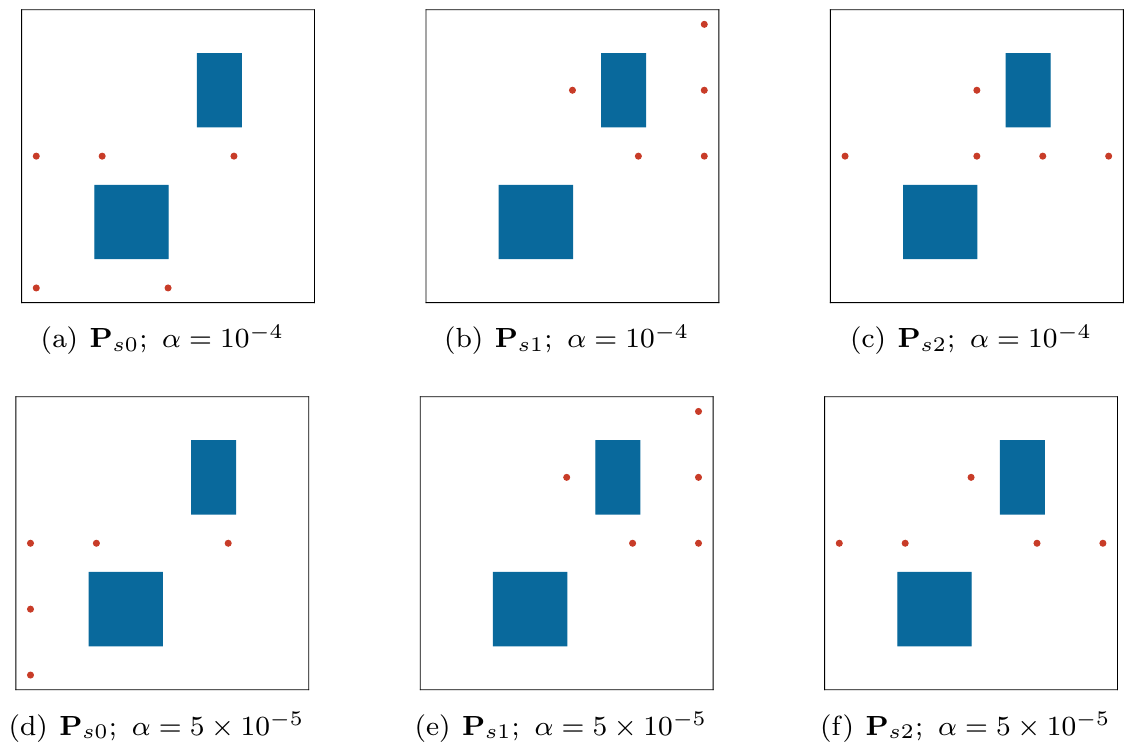}
		\caption{GOODE results for the model problem (Section~\ref{sec:Experiment_Setup}), which are the thresholded results of Figure~\ref{fig:scalar_GOOED_A_Optimal_weights_3d}.
		The weights are thresholded such that the sensors with highest $k=5$ weights are selected.
		}
		\label{fig:scalar_GOOED_A_Optimal_weights_upper_quartile}
	\end{figure}

	It is obvious that different choices of the penalty parameter $\alpha$ can generally lead to different optimal designs.
	Intuitively speaking, increasing the value of $\alpha$ is expected to result in a more sparsified design.
	See Section~\ref{subsec:reg_parameter} for a discussion on the choice of the design penalty $\alpha$.

	\subsection{Vector-valued goal}\label{subsec:vector_goal_results}
This section presents numerical experiments of the A-GOODE and D-GOODE
problems with vector-valued prediction QoI defined as the concentration of the
contaminant within distance $\epsilon$ (see
Table~\ref{table:prediction_dimensions}) from the building(s) boundaries.

	%
	\subsubsection{Goal-oriented A-optimality}
	%
	
	%
	%
	\begin{figure} \centering
		%
    \includegraphics[width=0.73\linewidth]{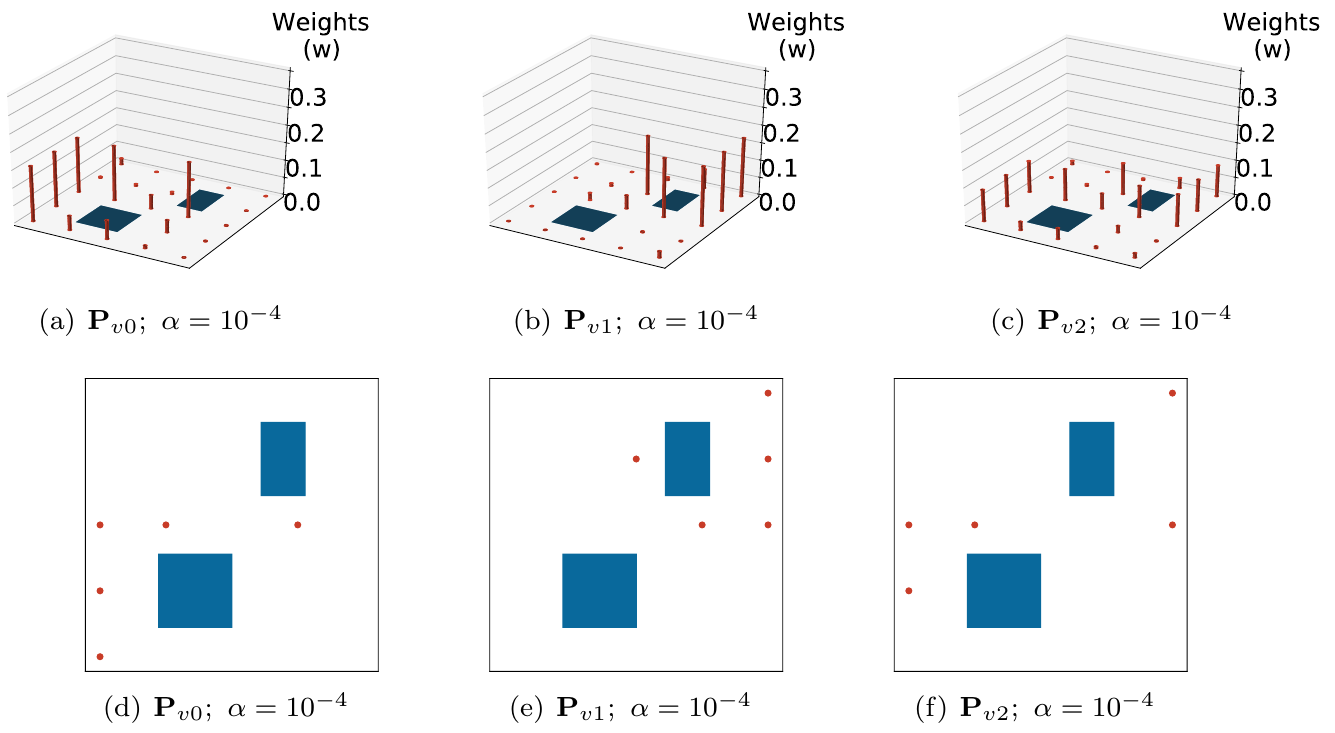}

		\caption{A-GOODE results for the model problem
(Section~\ref{sec:Experiment_Setup}).  
The prediction
operator, and the penalty parameter $\alpha$, are indicated under each panel.
The optimal weights $\{w_i\}_{i=1,\ldots, \Nsens}$ are plotted on the z-axis,
where the weights are normalized to add up to $1$ (top row); the corresponding active sensors are plotted on the bottom row.  }
\label{fig:vector_GOOED_A_Optimal_weights_3d} \end{figure}

Figure~\ref{fig:vector_GOOED_A_Optimal_weights_3d} shows the A-GOODE optimal
design, i.e., the optimal weights $\{w_i\}_{i=1,\ldots, \Nsens}$ of the
$\Nsens=22$ candidate sensor locations, for the penalty
parameter $\alpha = 10^{-4}$.  The thresholded solutions with $k=5$ target sensor
locations, are shown in the bottom row.  Results in
Figures~\ref{fig:scalar_GOOED_A_Optimal_weights_3d}
and~\ref{fig:vector_GOOED_A_Optimal_weights_3d} suggest that changing the
dimension of the predicted QoI, yields very different optimal designs.  This is
also supported by the results of the thresholded optimal designs in
Figures~\ref{fig:scalar_GOOED_A_Optimal_weights_3d}.
	\subsubsection{Goal-oriented D-optimality}
	%
	The D-GOODE problem with vector-valued prediction operators $\Predmat_{\rm vi};\, i = 0,\,1,\,2\,$, is solved for a sequence of $50$ values of the penalty parameter $\alpha$ evenly spaced between $[0.01,\,  0.5]$.
	Figure~\ref{fig:vector_GOOED_D_Optimal_weights_3d} (top row) shows the D-GOODE optimal design, i.e., the optimal weights $\{w_i\}_{i=1,\ldots, \Nsens}$ of the $\Nsens=22$ candidate sensor locations,
	for several choices of the penalty parameter $\alpha$. The corresponding thresholded solutions with $k=5$ target sensor locations, are shown in same figure (bottom row).
	While both the A-GOODE and D-GOODE designs are sensitive to the choice of $\alpha$, the results revealed in both Figures~\ref{fig:vector_GOOED_A_Optimal_weights_3d}, and~\ref{fig:vector_GOOED_D_Optimal_weights_3d}, strongly suggest that changing the GOODE criterion is expected to lead to a different design, even though the same prediction operators are used. This is again confirmed by comparing the thresholded A-GOODE and D-GOODE designs shown in Figures~\ref{fig:vector_GOOED_A_Optimal_weights_3d},~\ref{fig:vector_GOOED_D_Optimal_weights_3d} .

	\begin{figure}[!ht]
		\centering
    \includegraphics[width=0.73\linewidth]{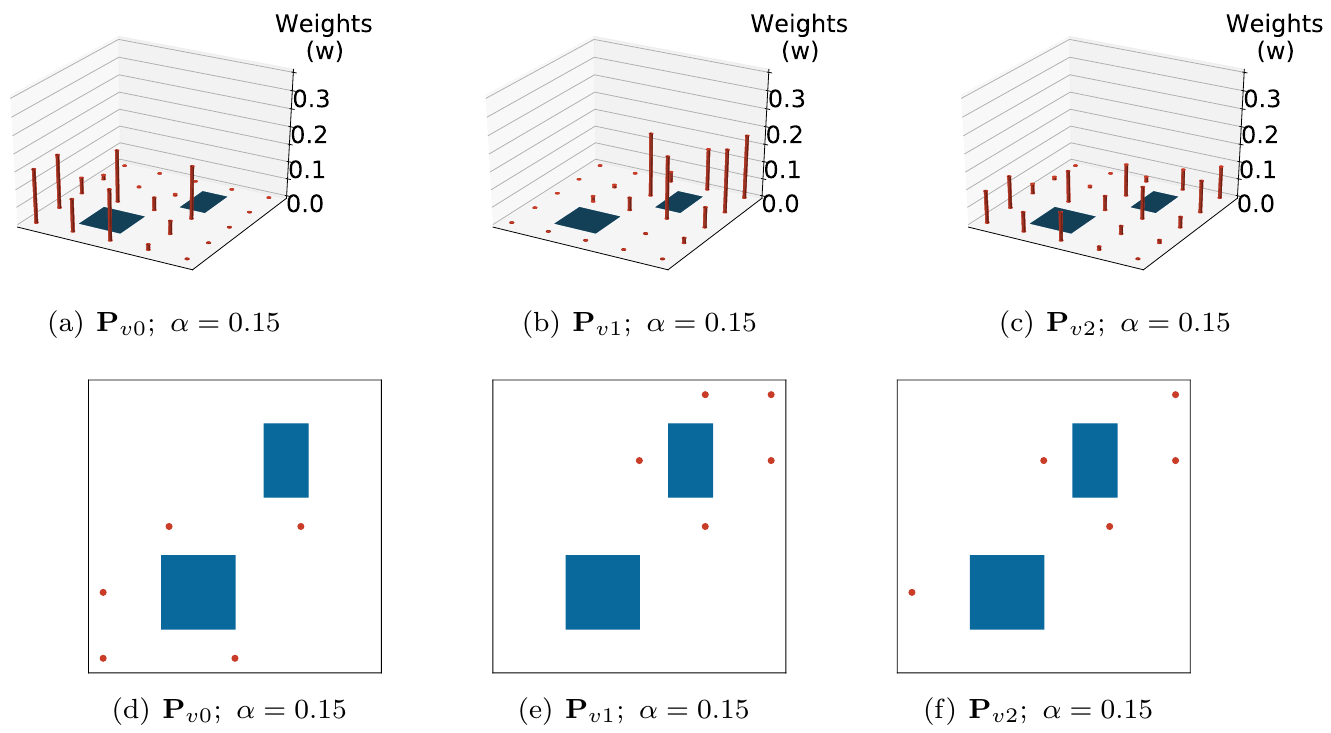}
		\caption{D-GOODE results for the model problem (Section~\ref{sec:Experiment_Setup}).
		The prediction operator, and the penalty parameter $\alpha$, are indicated under each panel.
		The optimal weights $\{w_i\}_{i=1,\ldots, \Nsens}$ are plotted on the z-axis, where the weights are normalized to add up to $1$; the thresholded versions are provided in the bottom row.
		}
		\label{fig:vector_GOOED_D_Optimal_weights_3d}
	\end{figure}
	%
\subsection{On the choice of the penalty parameter}\label{subsec:reg_parameter}


We now address the choice of the penalty parameter $\alpha$. The choice is
guided as a way to balance two opposing considerations: posterior predictive
uncertainty and number of active sensors. On the one hand, minimizing the
uncertainty is possible by collecting as much information as possible which is
achieved by keeping all the sensor weights active. On the other hand, because
of various constraints, the number of active sensors must be limited. This is
similar, in spirit, to the L-curve approach~\cite{hansen1999curve} for choosing
regularization parameters in the context of Tikhonov
regularization~\cite{phillips1962technique,tikhonov1963solution}. Our strategy
to determine a suitable value of $\alpha$ is as follows: for different values
of $\alpha$, plot the values of GOODE criterion $\Psi(\design)$ and the
sparsity $\|\design\|_1$ to form the L-shaped curve, and pick the value of
$\alpha$ which maximizes the curvature of the L-curve.

\begin{figure}[!ht]
	\centering
  \includegraphics[width=0.62\linewidth]{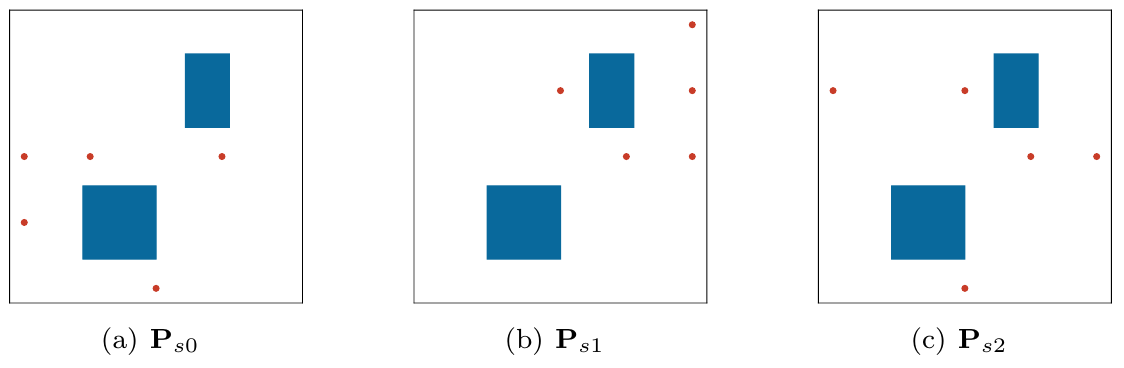}
	\caption{Brute-force optimal solution for GOODE optimality problem with the model problem (Section~\ref{sec:Experiment_Setup}).
		The number of sensors requested is set to $k=5$, out of $22$ candidate location.
		The prediction operator is indicated under each panel.
		}
	\label{fig:brute_force_scalar_A}
\end{figure}

To validate our approach for selecting the penalty parameter, we
performed a brute-force search of the design space to find the most efficient
design with $k=5$ active sensors out of $\Nsens=22$ candidate locations.
All-together, we evaluated $26{,}334$ designs and we found the optimal designs
for the scalar GOODE criterion (with prediction operators $\Predmat_{\rm si},\,
i = 0,\,1,\,2\,$). These results are displayed in
Figure~\ref{fig:brute_force_scalar_A}. We call these solutions as the
brute-force solutions. For a fair comparison with the solution of the
optimization problem~\eqref{eqn:optim_optimization_vector}, we first threshold
the solutions using the procedure described in Section~\ref{sec:sensors}. This
ensures that both the brute-force solution and the GOODE solution have the same
number of sensors.

\begin{figure}[!ht]
	\centering
    \includegraphics[width=0.62\linewidth]{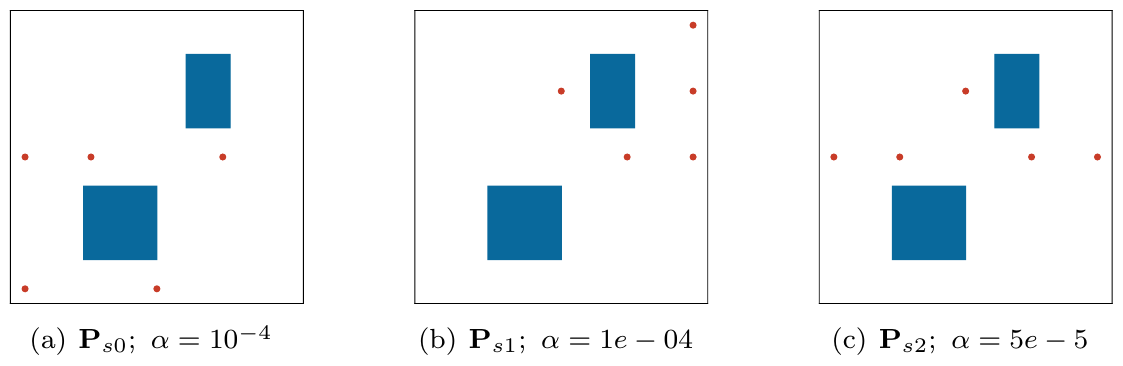}
\caption{Results of the scalar GOODE problem for the model problem (Section~\ref{sec:Experiment_Setup}), with penalty values resulting in best solution amongst values tested
from a design penalty pool. The goal is determined by the prediction operator $\Predmat_{\rm si},\, i = 0,\,1,\,2\,$.   The number
of sensors requested is set to $k=5$.  The prediction operator, and the
design penalty are indicated under each panel.  }
\label{fig:best_solutions_scalar_A}
\end{figure}

	\begin{figure}[!ht]
		\centering
    \includegraphics[width=0.90\linewidth]{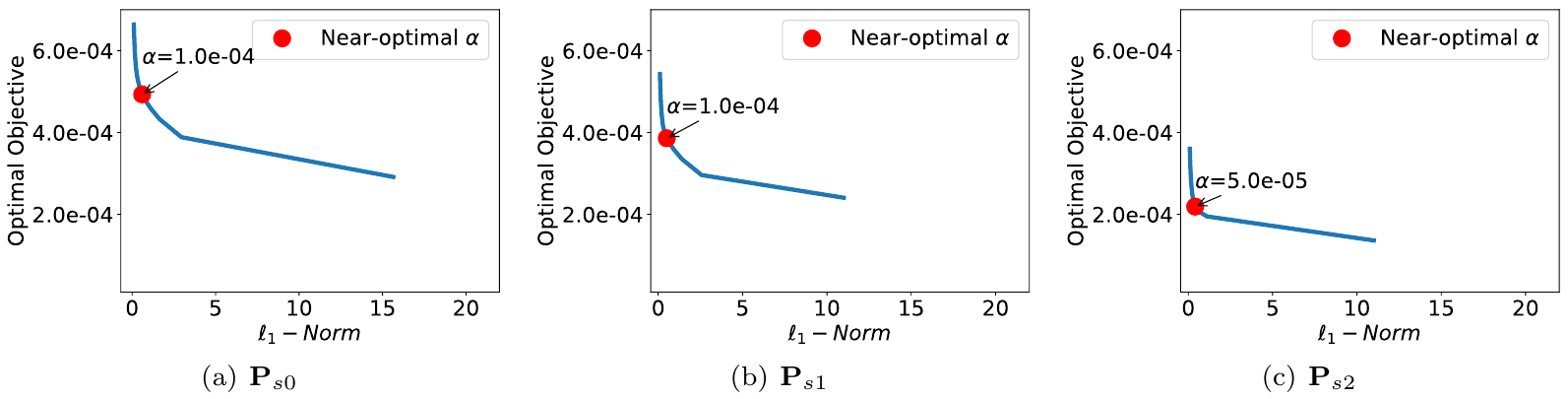}
		\caption{Results of the scalar GOODE problem with a sequence of $50$ penalty parameter values evenly spaced between $[10^{-7},\,  10^{-3}]$.
		The near-optimal penalty, is the value of $\alpha$ that leads to a thresholded design which is closest to the optimal solution obtained by brute force search.
		The near-optimal penalty is highlighted in each panel as a red dot, and the prediction operator used is indicated under each panel.
		}
		\label{fig:scalar_GOOED_regularization}
	\end{figure}

For each prediction operator $\Predmat_{\rm si},\, i = 0,\,1,\,2\,$, we
obtained the GOODE solution for a sequence of $50$ penalty parameter values
evenly spaced between $[ 10^{-7},\,  10^{-3}]$. 
Then we found the GOODE solution whose corresponding GOODE objective value 
was closest to the GOODE objective value at the brute-force solution.
%
%
The corresponding value of
$\alpha$ is called ``near-optimal,'' and is overlaid on the L-curve, see
Figure~\ref{fig:scalar_GOOED_regularization}. It is readily observed that the
``near optimal'' values of $\alpha$ are close to the regions of high curvature,
thereby providing empirical evidence for our strategy to pick a penalty
parameter. The ``near optimal'' values of $\alpha$ for the prediction operators
were found to be $\{10^{-4},\, 10^{-4},\, 5^{-5}\}$. As a further confirmation,
we plot the histogram of the GOODE objective function values at 
the $26{,}334$ possible designs. The design that gave the smallest objective
function was unique. Moreover the GOODE design computed with the near optimal
$\alpha$ had an objective function that is close to the brute force solution.

\begin{figure}[!ht]
	\centering
    \includegraphics[width=0.90\linewidth]{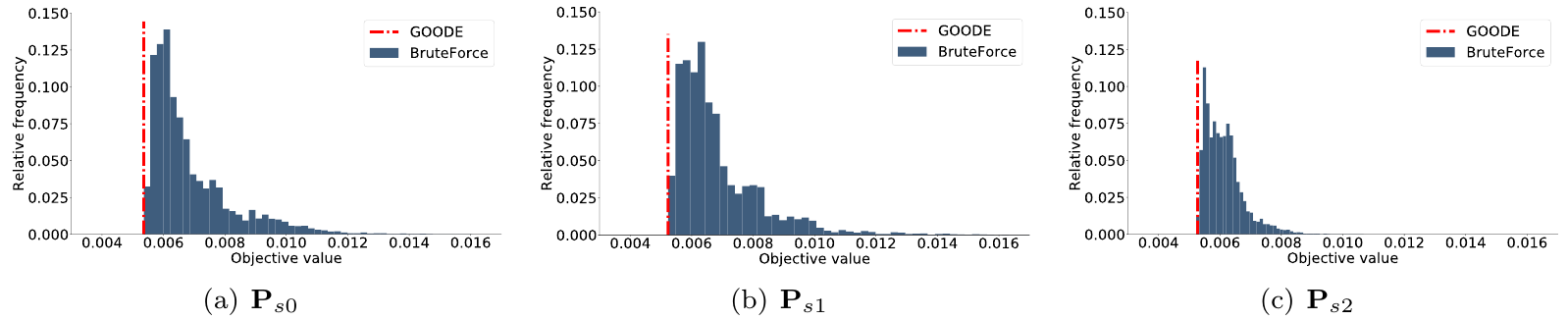}
	\caption{Scalar GOODE results for the model problem (Section~\ref{sec:Experiment_Setup}).
	The goal is determined by the prediction operators $\Predmat_{\rm si},\, i = 0,\,1,\,2\,$.
	The plots show the histogram of the objective function $\Psi(\design)$ evaluated for the $26{,}334$ possible binary designs with active $k=5$ sensors, out of $\Nsens=22$ candidate sensor locations.
	The value of the objective function for the thresholded GOODE optimal design obtained by the near-optimal values of $\alpha$ is shown as a vertical dashed red line.
	The prediction operator, and the design penalty are indicated under each panel.
	}
	\label{fig:scalar_GOOED_objective_freq}
\end{figure}

We repeat the same analysis for the A-GOODE case, where the objective is to
predict the concentration (vector) of the contaminant around one or two
buildings---that is using the vector prediction operators $ \Predmat_{\rm
vi},\, i = 0,\,1$.  Specifically, in
Figure~12(a) and~12(b),
we plot the L-curve for $75$
values of the penalty parameter $\alpha$ in the interval $[ 10^{-7},\,  0.2]$
where the prediction operators $\Predmat_{\rm vi},\, i=0,\,1$ are used. The
conclusions are similar: the near optimal parameter $\alpha$ is close to the
elbow of the L-curve, and the histogram in Figure~\ref{fig:A_GOOED_objective_freq_2} confirms this.

	\begin{figure}
		\centering
    \includegraphics[width=0.62\linewidth]{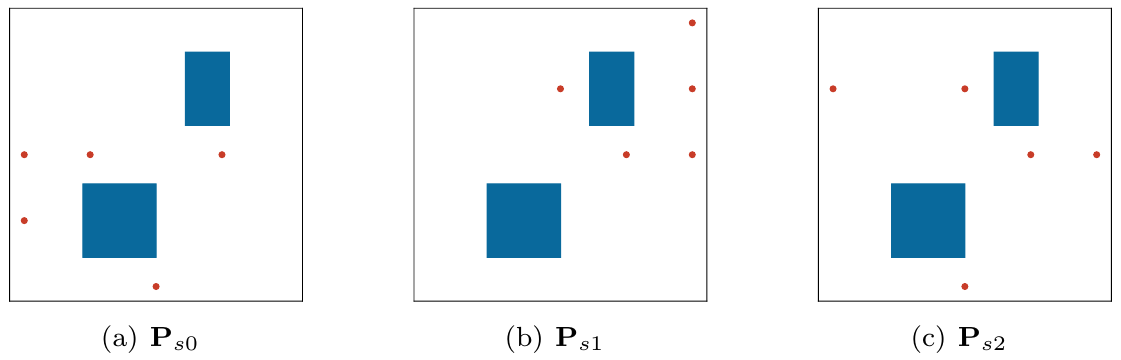}
		\caption{Brute-force optimal solution for A-GOODE optimality problem with the model problem (Section~\ref{sec:Experiment_Setup}).  
		The number of sensors requested is set to $k=5$, out of $22$ candidate sensor locations.
		The prediction operator is indicated under each panel.
		}
		\label{fig:brute_force_vector_A}
	\end{figure}
	%
	%


	%
	%
	\begin{figure}[!ht]
		\centering
    \includegraphics[width=0.64\linewidth]{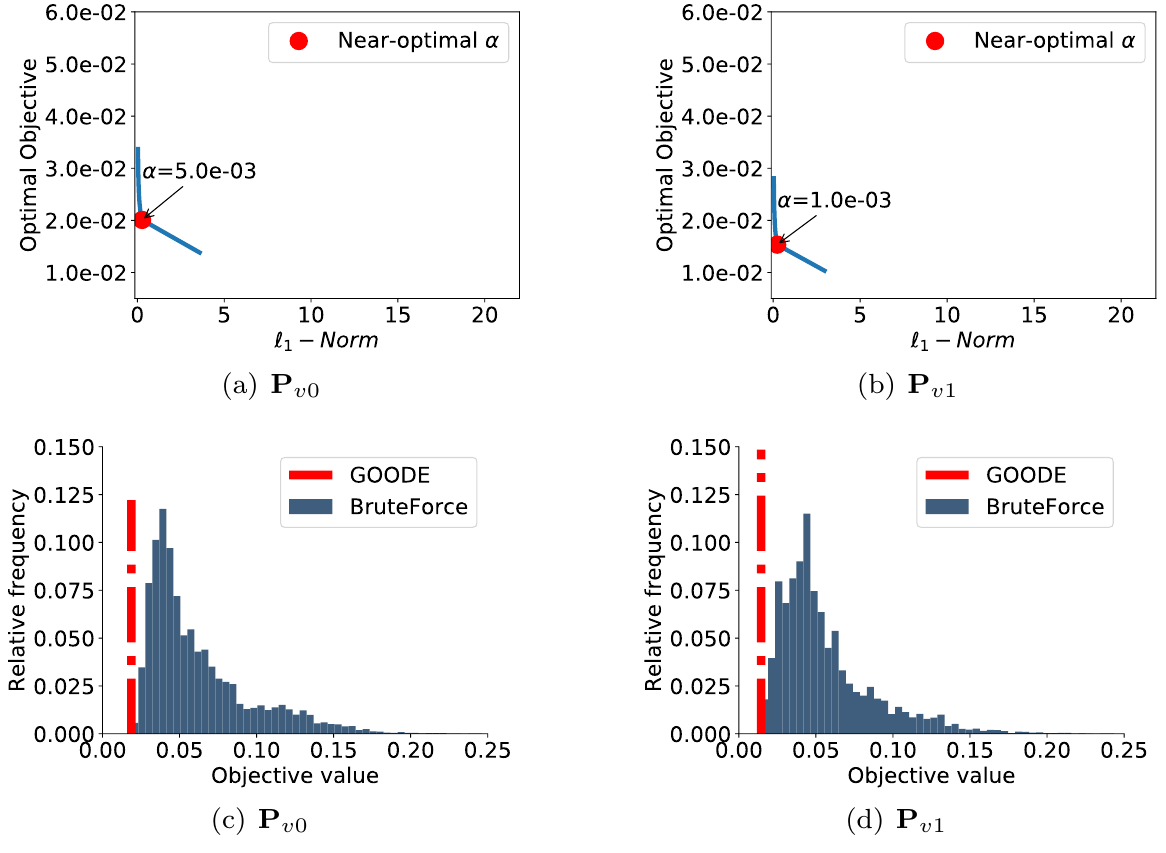}
		\caption{A-GOODE results for the model problem
(Section~\ref{sec:Experiment_Setup}).
    Panels~12(a), and~12(b)
show results of the A-GOODE
problem with a sequence of $75$ (equally spaced) penalty parameter values in
the interval $[10^{-7},\,  0.2]$.  The near-optimal penalty, is the value of
$\alpha$, that leads to a thresholded design which is closest to the brute
force solution and is highlighted in each panel as a red dot, and the
prediction operator used is indicated under each panel.
Panels~12(c), and~12(d)
show the distributions of the values of the objective function evaluated for
the $26{,}334$ possible binary designs with active $k=5$ sensors, out of
$\Nsens=22$ candidate sensor locations, for the three prediction operators. The red
line denotes the objective function value obtained at the near-optimal value of
$\alpha$.  The prediction operator and the design penalty are indicated under
each panel.  }
		\label{fig:vector_A_GOOED_regularization}
	\end{figure}
	\begin{figure}[!ht]
		\centering
    \includegraphics[width=0.64\linewidth]{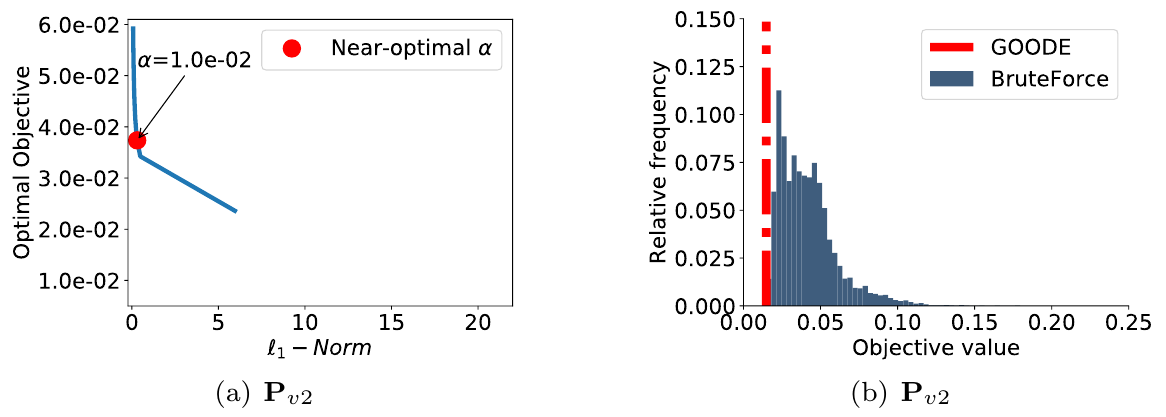}
		\caption{Same setup as Figure~\ref{fig:vector_A_GOOED_regularization}, with a prediction operator $\Predmat_{\rm v2}$. }
		\label{fig:A_GOOED_objective_freq_2}
	\end{figure}
	%
	%


\section{Conclusion}\label{sec:Conclusions}
%
We have presented mathematical and algorithmic foundations for goal-oriented
optimal design of experiments (GOODE), for PDE-based Bayesian linear inverse
problems.  The presented formulations provide natural extensions to the
classical Bayesian A- and D-optimal experimental design. Our theoretical
analysis of the presented goal-oriented criteria provides a clear decision
theoretic understanding of the A-GOODE and D-GOODE criteria in terms of
expected Bayes risk and expected information gain for the estimated 
end-goal quantities of interest, 
respectively. Our gradient based optimization framework provides a
concrete numerical recipe for computing A- and D-GOODE experimental designs.
The sparsity of the designs is enforced via an $\ell_1$-penalty approach. The
presented numerical results, and a validation using exhaustive enumeration,
provide compelling evidence of the merits of the presented approach. Natural
extensions of the presented work includes practical GOODE strategies for
nonlinear Bayesian inverse problems and nonlinear end-goal operators.

	\section*{Appendix: Proofs and Derivations}  
	\begin{appendices}
		\section{Proof of Theorem~\ref{thm:bayesrisk}}\label{apdx:proof_of_bayesrisk}

		\newcommand{\avey}[1]{\Expect{\obs |\iparam}{ {#1} }}
		\begin{proof}
			First, note that
%

			\begin{equation}\label{eqn:MSE_decomp}
				 \begin{aligned}
				{ \Expect{\obs |\iparam}{\sqnorm{\preda^\obs - \pred}} }
				= & \>\Expect{\obs |\iparam}{\sqnorm{\Predmat \iparama - \Predmat \iparam}}  \\
				= & \>  \Expect{\obs |\iparam}{\sqnorm{\Predmat \iparama -  \Predmat \Expect{\obs |\iparam}{ \iparama} } }
				+  \sqnorm{ \Predmat \Expect{\obs |\iparam}{ \iparama} - \Predmat \iparam }, 
				  \end{aligned}
			\end{equation}
which is the usual bias variance decomposition~\cite{HaberHoreshTenorio08} for
$\preda^\obs$.  
Focusing on the first term in the above sum, and letting 
$\iparam$ be fixed but arbitrary, we note that
			\begin{equation*}
				\begin{aligned}
					\Expect{\obs |\iparam}{ \iparama}
					= & \>  \Expect{\obs |\iparam}{  \Cparampostmat
					\left( \Cparampriormat^{-1} \iparb + \F^* \Cobsnoise^{-1} \obs  \right)  }
					=  \Cparampostmat \left( \Cparampriormat^{-1} \iparb + \F^* \Cobsnoise^{-1}
					\Expect{\obs |\iparam}{\obs}  \right)    \\
					= & \> \Cparampostmat \left( \Cparampriormat^{-1} \iparb + \F^* \Cobsnoise^{-1} \F\iparam  \right)    
					=   \Cparampostmat \left( \Cparampriormat^{-1} \iparb + \HMmat \iparam \right).
				\end{aligned}
			\end{equation*}
Next, let us define the random variable ${S}(\obs) =
\Predmat(\iparama -  \Expect{\obs |\iparam}{ \iparama})$. 
It is straightforward to see that
${S}(\obs) = \Predmat\Cparampostmat \F^* \Cobsnoise^{-1} \left( \obs - \F \iparam \right)$.
Therefore, since $\obs \sim \GM{\F \iparam}{\Cobsnoise}$, 
${S}(\obs)$ has a Gaussian distribution law $\GM{\vec{0}}{\mat{\Gamma}_S}$
with 
$\mat{\Gamma}_S := 
(\Predmat\Cparampostmat \F^* \Cobsnoise^{-1})\Cobsnoise 
(\Predmat\Cparampostmat \F^* \Cobsnoise^{-1})^*
= \Predmat\Cparampostmat \HMmat \Cparampostmat \Predmat^*$.
Hence,
\begin{equation}\label{eqn:expected_post_pred_dev}
				\begin{aligned}
					\Expect{\obs |\iparam}{\sqnorm{\Predmat \iparama -  \Predmat \Expect{\obs |\iparam}{ \iparama} } }
					&= \int \sqnorm{ {S}(\obs)} \, \GM{\F \iparam}{\Cobsnoise} (d\obs)\\
					&= \int \sqnorm{ \vec{z}} \, \GM{\vec{0}}{\mat{\Gamma}_S} (d\vec{z}) = \trace(\mat{\Gamma}_S) 
= \trace(\Predmat\Cparampostmat \HMmat \Cparampostmat \Predmat^*).
				\end{aligned}
			\end{equation}
Next, we consider the second term in~\eqref{eqn:MSE_decomp}.  
A similar calculation shows,
\begin{equation}\label{eqn:second_term}
\Expect{\priorm}{ \sqnorm{ \Predmat \Expect{\obs |\iparam}{ \iparama} - \Predmat \iparam } }
= \trace(\Predmat \Cparampostmat \Cparampriormat^{-1} \Cparampostmat \Predmat^*).
\end{equation}
Thus, combining~\eqref{eqn:expected_post_pred_dev},~\eqref{eqn:second_term},
along with~\eqref{eqn:MSE_decomp}, we get
\begin{equation*}
\Expect{\priorm}{\Expect{\obs |\iparam}{\sqnorm{\preda^\obs - \pred}}}
= \trace\big(\Predmat\Cparampostmat(\HMmat +  \Cparampriormat^{-1}) \Cparampostmat \Predmat^*\big)
= \trace(\Predmat\Cparampostmat\Predmat^*).~\qedhere
\end{equation*}
\end{proof}

\section{Gradient Derivation of Goal-Oriented A-Optimality Objective}\label{appen:A_optim_discrete_vector_gradient}
In this section we derive the derivative of the goal-oriented A-optimality objective $\Psi^\GA$.
We take the derivative of trace of the end-goal posterior covariance matrix $\wCpredpostmat$ with respect to the design weights $w_i$, i.e.
		\begin{equation}
			\frac{\partial}{\partial\, w_i}  \Psi^\GA
			= \frac{\partial}{\partial\, w_i}  \Trace{ \wCpredpostmat }
			= \Trace{  \frac{\partial}{\partial\, w_i}   \wCpredpostmat   };\quad i=1,2,\ldots, \Nsens \,.
		\end{equation}
		Using $\Hessmat^{-1}(\design)  = \Cparampriormat^{-1} + \wHMmat$
		\begin{equation}\label{eqn:initial_trace_derivative}
			\, \Trace{\frac{\partial}{\partial\, w_i}  \wCpredpostmat }
			= \Trace{  \Predmat \,\frac{\partial}{\partial\, w_i}   \Hessmat^{-1}(\design) \, \Predmat^*  }
			= -\Trace{  \Predmat\Hessmat^{-1}(\design) \,\frac{\partial \wHMmat}{\partial\, w_i} \,\Hessmat^{-1}(\design) \Predmat^* }.
		\end{equation}
The derivative of the weighted Hessian misfit $\wHessmat$ with respect to the design $\design$, is obtained as follows:
		\begin{equation} \label{eqn:Hessian_misfit_derivative}
			\begin{aligned}
				\frac{\partial\, \wHMmat}{\partial\, w_i}  &= \frac{\partial\, }{\partial\, w_i} \Bigl( \F^* \wdesignmat \F \Bigr) 
				= 	\F^*\, \Cobsnoise^{-1/2} \left( \frac{\partial\, \designmat}{\partial \, w_i}\right) \Cobsnoise^{-1/2} \,\F  \\
				&= \F^*\, \Cobsnoise^{-1/2} \left( \one \otimes \vec{e}_i \right) \left( \one \otimes \vec{e}_i \right)\tran \Cobsnoise^{-1/2} \,\F,  \\
			\end{aligned}
		\end{equation}
where $\one \in \mathbb{R}^{\Nobs}$ is a vector of ones, and $\vec{e}_i$ is the $i^{th}$ coordinate vector in $\Rnum^{\Nobs}$. Here $\otimes$ refers to the Kronecker product.

Given a set of temporally uncorrelated observations $ \{
\yt{k}=\y[\tind{k}] \}_{k=1,2,\ldots,\nobs}$, available at the discrete time
instances $\{ \tind{k} \}_{k=1,2,\ldots,\nobs} \subset [\tind{0},T_{\rm F}]$,
the derivative~\eqref{eqn:Hessian_misfit_derivative} expands to:
		\begin{equation}
			\frac{\partial\, \wHMmat}{\partial\, w_i}  = \sum_{k=1}^{\nobs}{  \F_{0,k}^*\, \Cobsnoisemat_k^{-\frac{1}{2}} \vec{e}_i \vec{e}_i\tran \Cobsnoisemat_k^{-\frac{1}{2}} \,\F_{0,k}  } \,,
		\end{equation}
where $\F_{0,k}$ is the forward model that maps the parameter
to the equivalent sensor measurements at time instance $\tind{k}$, and
$\Cobsnoisemat_k$ is the covariance of the measurement noise at time instance
$\tind{k}$.

From~\eqref{eqn:initial_trace_derivative} and
\eqref{eqn:Hessian_misfit_derivative}), it follows that:
		\begin{equation} \label{eqn:posterior_cov_trace_derivative}
			\begin{aligned}
				\frac{\partial\, }{\partial\, \design_i} \Trace{ \wCpredpostmat} 	&= - \trace{ \left( \Predmat \, \left[ \wHessmat \right]^{-1} \sum_{k=1}^{\nobs}{  \F_{0,k}^*\, \Cobsnoisemat_k^{-\frac{1}{2}} \vec{e}_i \vec{e}_i\tran \Cobsnoisemat_k^{-\frac{1}{2}} \,\F_{0,k}  } \left[\wHessmat\right]^{-1} \Predmat^*  \right)
				}  \\
				&= -\sum_{k=1}^{\nobs}{ \sum_{j=1}^{\Npred} \left( \Cobsnoisemat_k^{-\frac{1}{2}} \F_{0,k}  \, \left[\wHessmat\right]^{-1} \Predmat^* \right)^2 _{i,j} }\qquad i=1,2,\ldots, \Nsens \,,
			\end{aligned}
		\end{equation}
where the subindex $i,j$ indicates the $i^{th}$ row, and $j$ is
the $j^{th}$ column of a matrix.
Some elementary algebra shows  
\begin{equation} \label{eqn:A_optim_discrete_vector_gradient_derived}
			\begin{aligned}
				\nabla_\design \Trace{\wCpredpostmat} 		& = - \sum_{k=1}^{\nobs}{
				\sum_{j=1}^{\Npred}{
					\boldsymbol{\zeta}_{k,j} \odot \boldsymbol{\zeta}_{k,j}.
				}
				},
			\end{aligned}
		\end{equation}
where $\boldsymbol{\zeta}_{k,j}$ was defined in~\eqref{eqn:zetakj}.

		\section{Gradient derivation of goal-oriented D-Optimality Objective} \label{appen:D_optim_discrete_vector_gradient}
		In this section we derive the gradient of the goal-oriented D-optimality objective $\Psi^\GD$. The important observation is that
		\begin{equation}
			\frac{\partial\, }{\partial\, w_i} \Psi^\GD = 	\frac{\partial\, }{\partial\, w_i} \left( \log{ \det {\left(  \wCpredpostmat  \right)}} \right)
				= \trace{\left(  \Cpredpostmat^{-1}(\design)\, \frac{\partial\, \wCpredpostmat}{\partial\, w_i}   \right)}.
		\end{equation}

\subsection{First form}\label{appen:D_optim_discrete_vector_gradient_form1}

A similar argument as in the previous subsection gives for $i=1,\dots,\Nsens$ 
\begin{equation}\label{eqn:inter} 
	\frac{\partial\, \Psi^\GD}{\partial\, w_i}  = - \sum_{k=1}^{\nobs}\trace{ \left(\Cpredpostmat^{-1}(\design) \Predmat \, \left[ \wHessmat \right]^{-1}{  \F_{0,k}^*\, \Cobsnoisemat_k^{-\frac{1}{2}} \vec{e}_i \vec{e}_i\tran \Cobsnoisemat_k^{-\frac{1}{2}} \,\F_{0,k}  } \left[\wHessmat\right]^{-1} \Predmat^*  \right)}.
\end{equation}
Factor $\Cpredpostmat^{-1}(\design) = \Cpredpostmat^{-1/2}(\design) \Cpredpostmat^{-1/2}(\design)$, and use the cyclic property of trace operator to write 
\[\frac{\partial\, \Psi^\GD}{\partial\, w_i}  = - \sum_{k=1}^{\nobs}{ \sum_{j=1}^{\Npred} \left( \Cobsnoisemat_k^{-\frac{1}{2}} \F_{0,k}  \, \left[\wHessmat\right]^{-1} \Predmat^* \Cpredpostmat^{-1/2}(\design) \right)^2 _{i,j} }. \] 
With $\boldsymbol\xi_{k,j}$ defined in~\eqref{eqn:xikj}, we have the 
desired form~\eqref{eqn:D_optim_discrete_vector_gradient_form_1}.

\subsection{Second form}\label{appen:D_optim_discrete_vector_gradient_form2}

The second derivation takes a different turn at~\eqref{eqn:inter}. Use the cyclic property of the trace operator to rewrite as 
\begin{equation} 
	\frac{\partial\, \Psi^\GD}{\partial\, w_i}  = - \sum_{k=1}^{\nobs}{\left( \vec{e}_i\tran \Cobsnoisemat_k^{-\frac{1}{2}} \,\F_{0,k}   \left[\wHessmat\right]^{-1} \Predmat^*\Cpredpostmat^{-1}(\design) \Predmat \, \left[ \wHessmat \right]^{-1}  \F_{0,k}^*\, \Cobsnoisemat_k^{-\frac{1}{2}} \vec{e}_i   \right)}.
\end{equation}
The trace operator drops out since the summands are scalars. With
$\vec{\eta}_{k,i}$ as defined in~\eqref{eqn:etaki}, the derivative takes the
form of~\eqref{eqn:D_optim_discrete_vector_gradient_form_2}.

%
%
\section{On D-GOODE criterion}\label{appen:On_D_Optimality}
%
Recall that the goal operator $\Predmat$ maps elements of the discretized
parameter space, $\Rnum^{\Nparam}$, to elements of the end-goal space
$\Rnum^{\Npred}$.  The discretized parameter space is endowed with the inner
product $\ipg{\cdot}{\cdot}$ which, in general, is a discretization of the $L^2$
inner product. For instance, if finite-element discretization scheme is used,
this inner product will be the Euclidean inner product weighted with the finite
element mass matrix. The goal space $\Rnum^{\Npred}$ is endowed with the standard
Euclidean inner product.

The following lemma provides a key relation needed in proof
of the main result in this section.
\begin{lemma}\label{l_inter2}
With the definition  $\mat{G} = \Predmat^* { \left( \Predmat \Cparampriormat \Predmat^* \right) }^{-1}\Predmat$,
\begin{equation}\label{equ:Gmat}
\mathbb{E}_{\iparam}\, \mathbb{E}_{\obs|\iparam} \left[ \sqwnorm{ \preda^\obs - \predb  } {\Cpredpriormat^{-1} } \right] =  \trace(\Cparamprior\HMmat\Cparampost\mat{G}).
\end{equation}
\end{lemma}
\begin{proof}
With $\preda^\obs = \Predmat\iparpost$ and $\predb = \Predmat\iparprior$, we note
\[
   \sqwnorm{ \preda^\obs - \predb}{\Cpredpriormat^{-1}}
   = (\preda^\obs - \predb)^\top\Cpredpriormat^{-1}(\preda^\obs - \predb)
   = \ipg{\iparpost - \iparprior}{\mat{G}(\iparpost - \iparprior)}.
\]
Furthermore,
\begin{equation*}
\begin{aligned}
	\iparama - \iparamb
	= & \>\Cparampostmat \left( \Cparampriormat^{-1} \iparb + \F^* \Cobsnoise^{-1}\obs \right) - \iparamb\\
	= & \> \Cparampostmat \left( \Cparampriormat^{-1} -\Cparampostmat^{-1}\right) \iparb
           + \Cparampostmat \F^* \Cobsnoise^{-1} \obs
	= - \Cparampostmat \HMmat\iparb + \Cparampostmat \F^* \Cobsnoise^{-1}\obs.
\end{aligned}
\end{equation*}
We tackle the inner expectation in~\eqref{equ:Gmat} first.
Define $\vec{u} =  -\Cparampostmat\HMmat\iparb$ and $\mat{B}
= \Cparampostmat \F^* \Cobsnoise^{-1}$.
Using the fact $\obs | \iparam \sim
\GM{ \F \iparam}{\Cobsnoise}$, and the developments above,
\begin{equation}\label{equ:inner_expectation}
\begin{aligned}
\mathbb{E}_{\obs|\iparam} \left[ \sqwnorm{\preda^\obs - \predb}{\Cpredpriormat^{-1}} \right]
&=
\int \ipg{ \vec{u} + \mat{B}\obs }{\mat{G} (\vec{u} + \mat{B}\obs} \, \GM{ \F \iparam}{\Cobsnoise}(d\obs)\\
&=
\int \ipg{ \vec{v} }{\mat{G}\vec{v}} \, \GM{ \vec{u} + \mat{B}\F \iparam }{\mat{B}\Cobsnoise\mat{B}^*}(d\vec{v})\\
&=
\trace(\mat{B}\Cobsnoise\mat{B}^*\mat{G}) +
   \ipg{ \vec{z}}{\mat{G}\vec{z} } 
= \trace(\HMmat \Cparampostmat \mat{G}\Cparampostmat ) +   \ipg{\vec{z}}{\mat{G}\vec{z}},
\end{aligned}
\end{equation}
where, for short, we write $\vec{z} = \vec{u} + \mat{B}\F \iparam = \Cparampostmat \HMmat (\iparam-\iparb)$.

Next, we take the outer expectation over the prior distribution
$\mathcal{N}(\iparb,\Cparampriormat)$. Since the trace term is constant, 
we only need to consider the second term. A calculation 
similar to the one leading to~\eqref{equ:inner_expectation} shows
\begin{equation*}
\int
\ipg{\vec{z}}{\mat{G}\vec{z}}\,
\GM{\iparb}{\Cparampriormat}(d\iparam)
%
=\trace(\HMmat\Cparampostmat \mat{G} \Cparampostmat\HMmat \Cparampriormat).
\end{equation*}
%
%
%
%
Using this along with~\eqref{equ:inner_expectation}, 
we have
\begin{align*}
\mathbb{E}_{\iparam}\,
\mathbb{E}_{\obs|\iparam} \left[ \sqwnorm{\preda^\obs - \predb}{\Cpredpriormat^{-1}} \right]
&=
\trace(\HMmat\Cparampostmat \mat{G} \Cparampostmat) +
\trace(\HMmat\Cparampostmat \mat{G} \Cparampostmat\HMmat \Cparampriormat) \\
&= \trace(\HMmat\Cparampostmat \mat{G}\Cparampriormat),
\end{align*}
where the final equality follows from the identity $ \Cparampostmat(\mat{I} + \HMmat \Cparampriormat)=\Cparampriormat$.
\end{proof}


\begin{proof}[Proof of Theorem~\ref{thm:expinfogain}]
To show the equivalence, we compute the expected information gain. Since both
$\Pa(\pred|\obs)$ and $\Pb(\pred)$ are Gaussian,  Kullback--Leibler
divergence between these two distributions has an explicit expression given by
\begin{equation}\label{e_kld}
   \DKL{\Pa(\pred|\obs)}{\Pb(\pred)}
   = \frac12 \left[ \trace{ \left( \Cpredpriormat^{-1} \Cpredpostmat \right) }
      + \sqwnorm{ \preda - \predb  } {\Cpredpriormat^{-1} }    - {\Npred}
       - \log\det\Cpredprior^{-1}\Cpredpost \right].
\end{equation}
Using the cyclic property of the trace,
\begin{equation}\label{eqn:IG_term1}
\begin{aligned}
   \trace{ \left( \Cpredpriormat^{-1} \Cpredpostmat \right) }
   &= \trace{ \left( \left(\Predmat \Cparampriormat \Predmat^* \right)^{-1} \,
      \Predmat \Cparampostmat \Predmat^* \right) } =  \trace{ (\Cparampostmat\mat{G})}\,,
\end{aligned}
\end{equation}
where  $\mat{G} = \Predmat^* { \left( \Predmat \Cparampriormat \Predmat^* \right) }^{-1}\Predmat$.

To compute the expected information gain, note that the only term that depends
on the data and the prior is $\sqwnorm{ \wpreda - \predb  }
{\Cpredpriormat^{-1} }$. Using Lemma~\ref{l_inter2},
we have 
\begin{equation}\label{e_klde}
\begin{aligned}
\mathbb{E}_{\iparam}\, \mathbb{E}_{\obs|\iparam}\left(
\DKL{\Pa(\pred|\obs)}{\Pb(\pred)}\right)
= & \> \frac12 \left[ \trace(\Cparampost \mat{G}) +
  \trace(\Cparamprior \HMmat\Cparampost \mat{G} ) \right.\\
  & \qquad\qquad \left.- {\Npred} 
   - \log\det\Cpredprior^{-1}\Cpredpost\right].
\end{aligned}
\end{equation}
Only the last term survives as we now show. From $ \Cparampostmat(\mat{I} + \HMmat \Cparampriormat)=\Cparampriormat$, the first two terms simplify to 
\[ \trace(\Cparampost \mat{G}) +
  \trace(\Cparamprior \HMmat\Cparampost \mat{G} )= \trace(\Cparamprior \mat{G}) = \trace(\Cparamprior^{1/2} \mat{G}\Cparamprior^{1/2}) = \Npred.\]  
The final equality needs justification: it can be seen that $\Cparamprior^{1/2} \mat{G}\Cparamprior^{1/2}$ is an orthogonal projector which projects onto $\mathsf{range}(\Cparamprior^{1/2} \Predmat^*)$. Furthermore,   
note that the rank of an  
orthogonal projector equals its trace (this 
can be shown, for
example, using an SVD based argument). Since the first three terms
of~\eqref{e_klde} cancel, we are left with the last term which gives the
desired result.  
\end{proof}
%
Note that, this result can be generalized to the case where the end-goal operator does not have
full row-rank, for example by incorporating the formulation of the KL-divergence discussed in~\cite{attia2016reducedhmcsmoother}, where the covariance matrix is possibly rank deficient.

	\end{appendices}

\section*{Acknowledgments}
This material was based upon work partially supported by the NSF under Grant DMS-1127914 to the Statistical and Applied Mathematical Science Institute (SAMSI).

\bibliographystyle{unsrt}
 \section*{References}
	\bibliography{oed_references,data_assim_HMC}

\null
\vfill

\begin{flushright}
\scriptsize \framebox{\parbox{3.2in}{
The submitted manuscript has been created by UChicago Argonne, LLC,
Operator of Argonne National Laboratory (``Argonne"). Argonne, a
U.S. Department of Energy Office of Science laboratory, is operated
under Contract No. DE-AC02-06CH11357. The U.S. Government retains for
itself, and others acting on its behalf, a paid-up nonexclusive,
irrevocable worldwide license in said article to reproduce, prepare
derivative works, distribute copies to the public, and perform
publicly and display publicly, by or on behalf of the Government.
The Department of
Energy will provide public access to these results of federally sponsored research in accordance
with the DOE Public Access Plan. http://energy.gov/downloads/doe-public-access-plan. }}
\normalsize
\end{flushright}

\end{document}

%% file: notation.tex
\usepackage{xcolor}
\usepackage{amssymb}
\usepackage{mathrsfs}
\usepackage{amsthm}
\usepackage{moreverb}
\usepackage{upgreek}

 \expandafter\let\csname equation*\endcsname\relax 
 \expandafter\let\csname endequation*\endcsname\relax 
\usepackage{amsmath}

\usepackage{xspace}

\usepackage[utf8]{inputenc}
\usepackage[english]{babel}

\DeclareMathAlphabet      {\mathup}{OT1}{\familydefault}{m}{n}

\newtheorem{theorem}{Theorem}[section]

\newtheorem{lemma}[theorem]{Lemma}

\newtheorem{remark}[theorem]{Remark}

\xdefinecolor{samsi_green}{rgb}{0.04, 0.85, 0.32}
\xdefinecolor{samsi_darkgreen}{rgb}{0.24, 0.7, 0.44}
\xdefinecolor{samsi_mint}{rgb}{0.24, 0.71, 0.54}
\xdefinecolor{samsi_officegreen}{rgb}{0.0, 0.5, 0.0}
\xdefinecolor{samsi_napiergreen}{rgb}{0.16, 0.5, 0.0}
\xdefinecolor{samsi_maroon}{rgb}{0.65,0.06,0.17}
\xdefinecolor{samsi_blue}{rgb}{0,0.2,0.6}
\xdefinecolor{samsi_phthalogreen}{rgb}{0.07,0.21,0.14}
\definecolor{grassgreen}{RGB}{92,135,39}
\newcommand{\diag}[1]{\mathsf{diag}\left(#1\right)}                 
\newcommand{\blkdiag}[1]{\mathsf{blkdiag}\left(#1\right)}                 
\renewcommand{\mat}[1]{\mathbf{{#1}}}                                 
\renewcommand{\vec}[1]{{\mathchoice                                 
                       {\mbox{\boldmath$\displaystyle{#1}$}}
                       {\mbox{\boldmath$\textstyle{#1}$}}
                       {\mbox{\boldmath$\scriptstyle{#1}$}}
                       {\mbox{\boldmath$\scriptscriptstyle{#1}$}}}
                      }
\newcommand*{\tran}{^{\mkern-1.5mu\mathsf{T}}}                
\newcommand{\domain}{\mathcal{D}}                               
\newcommand{\trace}{\mathsf{tr}}                              
\newcommand{\Trace}[1]{\mathsf{tr} \left(#1\right)}                              
\newcommand{\norm}[1]{\left\| {#1} \right\|}                  
\newcommand{\sqnorm}[1]{\left\| {#1} \right\|^2}              
\newcommand{\sqwnorm}[2]{\left\| {#1} \right\|^2_{#2}}        
\newcommand{\ip}[2]{{\left\langle {#1}, {#2} \right\rangle}}  
\newcommand{\ipg}[2]{{\langle\!\langle {#1}, {#2} \rangle\!\rangle}}  
\newcommand{\wip}[3]{\left\langle{#1},                        
    {#2}\right\rangle_{\!\scriptscriptstyle{\mathup{#3}}}}
\newcommand\restr[2]{                                         
                     { \left.\kern-\nulldelimiterspace
                     {#1}\vphantom{\big|} \right|_{#2}}
                    }
\newcommand{\Rnum}{\mathbb{R}}  
\newcommand{\xcont}{u}                           

\newcommand{\y}{\mathbf{y}}                               
\newcommand{\obs}{\y}                                     
\newcommand{\yt}[1]{{\y}_{#1}}                            

\newcommand{\paramcont}{\theta}                               
\newcommand{\param}{\vec{\theta}}                               
\newcommand{\iparam}{\param}                              
\newcommand{\iparprior}{\iparam_{\rm pr}}            
\newcommand{\iparb}{\iparprior}            
\newcommand{\iparamb}{\iparprior}            
\newcommand{\iparpost}{\iparam_{\rm post}^\obs}           
\newcommand{\ipara}{\iparpost}            
\newcommand{\iparama}{\iparpost}            
\newcommand{\wiparama}{\iparpost(\design)}            
\newcommand{\Nstate}{\textsc{N}_{\rm state}}                
\newcommand{\Nparam}{{N_\mathup{\theta}}}
\newcommand{\Nobs}{\textsc{N}_{\rm obs}}                    
\newcommand{\Npred}{\textsc{N}_{\rm goal}}                  
\newcommand{\nobs}{{\textsc{N}_t}}                   
\newcommand{\Nsens}{\textsc{N}_{\rm s}}   

\newcommand{\Cparamprior}{\mat\Gamma_{{\rm pr}}}                
\newcommand{\Cparampost}{\mat\Gamma_{{\rm post}}}               
\newcommand{\Cobsnoise}{\mat{\Gamma}_{ {\!\rm noise}}}            
\newcommand{\Cpred}{\mat{\Sigma}}  
\newcommand{\Cpredprior}{\Cpred_{ {\rm pr}}}  
\newcommand{\Cpredpost}{\Cpred_{ \rm post}}   
\newcommand{\Cparampriormat}{\Cparamprior}                  
\newcommand{\Cparampostmat}{\Cparampost}                   
\newcommand{\wCparampostmat}{\Cparampostmat(\design)}                   

\newcommand{\Cobsnoisemat}{\mat{R}}                       

\newcommand{\Cpredpriormat}{\Cpredprior}  
\newcommand{\Cpredpostmat}{\Cpredpost}    
\newcommand{\wCpredpostmat}{\Cpredpostmat(\design)}    
\newcommand{\F}{\mathbf{F}}                                          
\newcommand{\Hessmat}{\mat{H}}                                       
\newcommand{\HMmat}{\Hessmat_{ {\rm misfit}}}                    
\newcommand{\Predmat}{\mat{P}}                                           
\newcommand{\GA}{\mathup{GA}}
\newcommand{\GD}{\mathup{GD}}
\newcommand{\KL}{\mathup{KL}}

\newcommand{\one}{\mathbf{1}}

\newcommand{\Prob}{\mathbb{P}}                                 
\newcommand{\GM}[2]{\mathcal{N}\!\left( {#1}, {#2}\right)}     
\newcommand{\Pb}{\Prob^{\rm b}}                                
\newcommand{\Pa}{\Prob^{\rm a}}                                

\newcommand{\priorm}{\mu_{\mathup{pr}}}                          
\newcommand{\DKL}[2]{\mathsf{D}_\mathup{KL}\left\{{#1}\, \|\, {#2}\, \right\}}  
\newcommand{\Expect}[2]{\mathbb{E}_{#1}{\left[ #2 \right]} }  
\newcommand{\tind}[1]{t_{\rm #1}}

\newcommand{\D}{\mathcal{D}}

\newcommand{\design}{\vec{w}}      
\newcommand{\designmat}{\mat{W}}   
\newcommand{\wdesignmat}{\designmat_{\Gamma}}   

\newcommand{\pred}{\vec{\rho}}                                    
\newcommand{\predb}{\pred_{\rm pr}}                                 
\newcommand{\preda}{\pred_{\rm post}}                                 
\newcommand{\wpreda}{\preda(\design)}                                 

\newcommand{\wHessmat}{\Hessmat(\design)}                    
\newcommand{\wHMmat}{\Hessmat_{\rm misfit}(\design)}      
\newcommand{\commentout}[1]{\iffalse {#1} \fi}

